\documentclass[11pt]{article}

\usepackage{fullpage}
\usepackage{amssymb,amsmath,amsfonts,amsthm,amssymb,mathrsfs}
\usepackage[usenames]{color}
\usepackage{hyperref}
\usepackage{xspace}
\usepackage[normalem]{ulem}

\definecolor{ToDoColor}{rgb}{0.1,0.2,1}

\usepackage{graphicx}
\usepackage{aliascnt,cleveref}
\newcommand{\remove}[1]{}
\renewcommand{\cref}{\Cref}

\newcommand{\calC}{{\cal C}}
\newcommand{\calI}{{\cal I}}
\newcommand{\notcalC}{\overline{{\cal C}}}
\newcommand{\calM}{{\cal M}}
\newcommand{\calX}{{\cal X}}
\newcommand{\calY}{{\cal Y}}
\newcommand{\calA}{\mathcal{A}}
\newcommand{\calP}{\Pi}
\newcommand{\crs}{w}
\newcommand{\CRS}{W}
\newcommand{\stepref}[1]{Step~(\ref{#1})}
\newtheorem{theorem}{Theorem}[section]
\crefname{theorem}{Theorem}{Theorems}

\crefname{notation}{Notation}{Notations}

\crefname{clm}{Claim}{Claims}
\newtheorem{claim}[theorem]{Claim}
\crefname{claim}{Claim}{Claims}

\crefname{proposition}{Proposition}{Propositions}
\newtheorem{lemma}[theorem]{Lemma}
\crefname{lemma}{Lemma}{Lemmas}
\newtheorem{definition}[theorem]{Definition}
\crefname{definition}{Definition}{Definitions}

\crefname{question}{Question}{Questions}

\newtheorem{corollary}[theorem]{Corollary}
\crefname{cor}{Corollary}{Corollaries}

\crefname{obs}{Observation}{Observations}
\theoremstyle{definition}
\newtheorem{remark}[theorem]{Remark}
\crefname{remark}{Remark}{Remarks}

\crefname{example}{Example}{Examples}
\crefname{construct}{Construction}{Constructions}
\crefname{fact}{Fact}{Facts}
\newcommand{\pprotocol}[5]{{\begin{figure}[#4]
\begin{center}
\fbox{
\hbox{\quad
\begin{minipage}{0.85\textwidth}
\begin{center}
{\bf #1}
\end{center}
\small
#5
\end{minipage}
} }
\caption{\label{#3} #2}
\end{center}
\end{figure} } }
\newcommand{\protocol}[4]{\pprotocol{#1}{#2}{#3}{htbp}{#4}}

\DeclareSymbolFont{AMSb}{U}{msb}{m}{n}
\DeclareMathSymbol{\NN}{\mathbin}{AMSb}{"4E}
\DeclareMathSymbol{\ZZ}{\mathbin}{AMSb}{"5A}
\DeclareMathSymbol{\RR}{\mathbin}{AMSb}{"52}
\DeclareMathSymbol{\QQ}{\mathbin}{AMSb}{"51}
\DeclareMathSymbol{\II}{\mathbin}{AMSb}{"49}
\DeclareMathSymbol{\CC}{\mathbin}{AMSb}{"43}
\renewcommand{\vec}[1]{\boldsymbol{#1}}
\DeclareMathOperator{\SD}{SD}
\DeclareMathOperator{\poly}{poly}
\DeclareMathOperator{\Bin}{Binomial}
\DeclareMathOperator{\view}{view}
\DeclareMathOperator{\Output}{Output}
\renewcommand{\epsilon}{\varepsilon}
\newcommand{\eps}{\varepsilon}

\newcommand{\calH}{\mathcal{H}}
\newcommand{\set}[1]{\left\{ #1 \right\}}
\newcommand{\floor}[1]{\left\lfloor #1 \right\rfloor}

\newcommand{\N}{N}
\DeclareMathOperator*{\E}{\mathbb{E}}
\newcommand{\ProtMessageTransmission}{Protocol {\sc SecureMessageTransmission}\xspace}
\newcommand{\ProtKeyExchange}{Protocol {\sc KeyExchange}\xspace}
\newcommand{\ProtMPCinShuffle}{Protocol {\sc MPCinShuffle}\xspace}
\newcommand{\ProtCommonPrelude}{Protocol {\sc CommonPrelude}\xspace}
\newcommand{\ProtCommonTwoRound}{Protocol {\sc CommonTwoRound}\xspace}
\newcommand{\ProtNestedCommon}{Protocol {\sc NestedCommonElement}\xspace}

\newcommand{\ProtABT}{Protocol {\sc ABT}\xspace}
\newcommand{\View}{{\mbox{\rm View}}}
\newcommand{\Sim}{{\mbox{\rm Sim}}}

\begin{document}
\title{On the Round Complexity of the Shuffle Model}
\author{Amos Beimel\thanks{Department of Computer Science, Ben-Gurion University. {\tt amos.beimel@gmail.com}}
\and Iftach Haitner\thanks{School of Computer Science, Tel-Aviv University. {\tt iftachh@cs.tau.ac.il}} 
\and Kobbi Nissim\thanks{Department of Computer Science, Georgetown University. {\tt kobbi.nissim@georgetown.edu}} 
\and Uri Stemmer\thanks{Department of Computer Science, Ben-Gurion University and Google Research. {\tt u@uri.co.il}}
}

\date{September 28, 2020}

\begin{titlepage}
\maketitle

\thispagestyle{empty}

\begin{abstract}
The shuffle model of differential privacy [Bittau et al.\ SOSP 2017; Erlingsson et al.\ SODA 2019; Cheu et al.\ EUROCRYPT 2019] was proposed as a viable model for performing distributed differentially private computations. Informally, the model consists of an untrusted analyzer that receives messages sent by participating parties via a shuffle functionality, the latter potentially disassociates messages from their senders. 
Prior work focused on one-round differentially private shuffle model protocols, demonstrating that functionalities such as addition and histograms can be performed in this model with accuracy levels similar to that of the curator model of differential privacy, where the computation is performed by a fully trusted party. A model closely related to the shuffle model was presented in the seminal work of Ishai et al.\ on establishing cryptography from anonymous communication [FOCS 2006].

Focusing on the round complexity of the shuffle model, we ask in this work what can be computed in the shuffle model of differential privacy with two rounds. Ishai et al.\ showed how to use one round of the shuffle to establish secret keys between every two parties. Using this primitive to simulate a general secure multi-party protocol increases its round complexity by one. We show how two parties can use one round of the shuffle to send secret messages without having to first establish a secret key, hence retaining round complexity. Combining this primitive with the two-round semi-honest protocol of Applebaun, Brakerski, and Tsabary [TCC 2018], we obtain that every randomized functionality can be computed in the shuffle model with an honest majority, in merely two rounds. This includes any differentially private computation.

We hence move to examine differentially private computations in the shuffle model that (i) do not require  the assumption of an honest majority, or (ii) do not admit one-round protocols, even with an honest majority. For that, we introduce two computational tasks: {\em common element}, and {\em nested common element with parameter $\alpha$}. For the common element problem we show that for large enough input domains, no one-round differentially private shuffle protocol exists with constant message complexity and negligible $\delta$, whereas a two-round protocol exists where every party sends a single message in every round.  For the nested common element we show that no one-round differentially private protocol exists for this problem with adversarial coalition size $\alpha n$. However, we show that it can 
be privately computed in two rounds against coalitions of size $cn$ for every $c<1$. This yields a separation between one-round and two-round protocols. We further show a one-round protocol for the nested common element problem that is differentially private with coalitions of size smaller than $c n$ for all $0<c<\alpha<1/2$.

\end{abstract}

\bigskip

{\bf Keywords:} Shuffle Model, Differential privacy, Secure Multiparty Computation
\end{titlepage}

\section{Introduction}

A recent line of work in differential privacy focuses on a distributed model where parties communicate with an analyzer via a random shuffle. The shuffle collects messages from the participating parties and presents them to the analyzer in a random order, hence potentially disassociating between messages and their senders~\cite{BittauEMMRLRKTS17,ErlingssonFMRTT19,CheuSUZZ19}. 
The hope is that the shuffle model would be useful for the implementation of practical distributed differentially private statistical and machine learning analyses, and with accuracy comparable to that of centralized differential privacy solutions. The implementation of the shuffle itself is envisioned to be based on technologies such as secure enclaves, mix nets, and secure computation.

The theoretical work on the shuffle model has so far focused on developing protocols for the model formalized in~\cite{CheuSUZZ19}. In this synchronous one-round model, all the participating parties send their messages through the shuffle at once (parties may send one message or multiple messages). Already in this limited communication model there are fundamental statistical tasks for which differentially private shuffle model protocols exist with error comparable to that achievable in the (centralized) curator model of differential privacy~\cite{CheuSUZZ19, BalleBGN19, GhaziGKPV19, BalleBGN19a, GhaziMPV20, BalleBGN20,BC20, GPVS19}.

A model similar to the shuffle model was presented already in 2006 by Ishai, Kushilevits, Ostrovsky, and Sahai in the context of secure multiparty computation~\cite{IKOS06}. 
In particular, Ishai et al.\ presented a one-round secure summation protocol that has become one of the building blocks of noise efficient real summation differentialy-private protocols, where each party holds a number  $x_i \in [0,1]$ and the analyzer's task is to estimate the sum $\sum x_i$~\cite{GhaziGKPV19, BalleBGN19a, GhaziMPV20, BalleBGN20}. Ishai et al.\ also presented a one-round protocol allowing any two parties to agree on a secret key, a step after which the parties can privately exchange messages. Combining this primitive with general constructions of secure multiparty computation protocols that rely on private or secure channels, Ishai et al.\ showed that it is possible to compute any (finite) function of the parties' joint inputs in a constant number of rounds. In particular, we observe that combining the key agreement protocol of Ishai et al.~\cite{IKOS06} with the recent two-round secure multiparty protocol of Applebaum, Brakersky, and Tsabary~\cite{ApplebaumBT18} (denoted the ABT protocol), no more than three rounds suffice for computing any (finite) randomized function securely in the shuffle model, with semi-honest parties assuming an honest majority: one round for every pair of parties to setup a secret key, and hence private communication channels. Two more round to simulate the ABT protocol using these private channels. To conclude, the previous results imply that any randomized function (including, in particular, any curator model differential privacy computation) can be computed in the shuffle model with security against an honest majority.\footnote{Curator model computations returning real numbers, such as those resulting by adding Laplace or Gaussian noise, would need to be carefully truncated to finite precision.}

\subsection{Our results}

In this work, we focus on the shuffle model with semi-honest parties. We ask what can be computed in the shuffle model with one and two rounds of communication, and at the presence of coalitions of semi-honest parties that can put together their inputs, randomization, and messages they receive during the computation with the goal of breaching the privacy of other parties. We present new techniques for constructing round-efficient protocols in the shuffle models as well as new lowerbound techniques for studying the limitations of one-round protocols. 
In more detail:

\medskip \noindent {\bf One-round private message transmission.} 
In \cref{sec:secureMessageTransmission} we present a new building block for shuffle model protocols. This is a protocol that allows a party $P_i$ to send a secret message to another party $P_j$ in one round. In the key agreement protocol of Ishai et al.~\cite{IKOS06}, mentioned above, to agree on a bit $b$ of the key, each of $P_i$ and $P_j$ selects and sends through the shuffle a random element chosen from a large set. Denoting the elements sent by $P_i, P_j$ as $x, y$ resp., parties $P_i$ and $P_j$ can set the secret bit $b$ to $0$ if $x<y$ and to $1$ if $x>y$. (The protocol fails if $x=y$.) The other parties cannot distinguish which of the two values is $x$ and which is $y$ and gain no information about the bit $b$. Using this protocol, party $P_i$ learns the secret key only after the conclusion of one communication round, and only then can $P_i$ use the key to encrypt a message. 
In contrast, our construction saves a round in the communication, as it allows $P_i$ to encrypt a message without having to first establish a key.

\smallskip \noindent {\bf Generic two-round secure multiparty computation for the shuffle model.}
Using the one-round message transmission protocol, we show in \cref{sec:twoRoundMPC} how to simulate the two-round  semi-honest secure multi-party computation protocol with information theoretic security of Applebaum et al.~\cite{ApplebaumBT18}.\footnote{
An alternative construction was given by Garg et al.~\cite{GargIS18}; the communication complexity of their protocol is exponential in the number of parties.}
The result is a general construction in the shuffle model of two-round honest majority protocols for the semi-honest setting, with information theoretic security. The construction is efficient in the size of the formula representing the functionality.

\smallskip

Our generic two-round construction shows that the shuffle model is extremely expressive: no more than two rounds suffice for computing any (finite) randomized function, including any curator level differential privacy computation, with semi-honest parties assuming an honest majority of players. We hence move to examine differentially private computations in the shuffle model that (i) do not require  the assumption of an honest majority, or (ii) do not admit one-round protocols, even with an honest majority. 
To demonstrate our lowerbound and upperbound techniques, we introduce two computational tasks:

\smallskip \noindent {\bf Common element:} Each of $n$ parties holds an input $x_i$ taken from a large finite domain $\calX$. The parties communicate with an analyzer via the shuffle. If all the parties hold the same input $x\in \calX$ then the analyzer's task is to output $x$. Otherwise, the analyzer's outcome is not restricted.

\smallskip \noindent {\bf Nested common element with parameter ${\boldsymbol{\alpha}}$:} This is a variant of the common element problem, where parties $P_1,\ldots,P_{\lfloor\alpha n\rfloor}$ each holds an input $x_i\in\calX$. 
The other parties $P_{\lfloor\alpha n\rfloor + 1},\ldots,P_n$ each holds a vector of $|\calX|$ elements taken from some finite domain $\calY$, i.e.,  $\vec{y}_i\in\calY^{|\calX|}$. The parties communicate with an analyzer via the shuffle. If all the parties of the first type hold the same input $x\in\calX$ and all the vectors held by parties of the second type have the same value $z$ in their $x$-th entry, then the analyzer's task is to output $z$ (otherwise, the analyzer's outcome is not restricted).
We consider the case where $|\calX|$ is polynomial in $n$, thus, the size of the inputs is polynomial in $n$ even when $|\calY|$ is exponential in $n$.

\medskip

Both tasks need to be performed with differential privacy, assuming semi-honest parties. We now describe the bounds we prove for these problems:

\medskip \noindent {\bf A lowerbound on one-round shuffle model protocols for the common element problem.} In  \cref{sec:lowerbound} we present a new lowerbound technique for one-round shuffle model protocols where the mutual information between input and output is high. Unlike other lowerbounds in the shuffle model of differential privacy that we are aware of, our lowerbound proof works for the multi-message setting, and does not require all parties to use the same randomizer.\footnote{Two exceptions are the recent works of Balcer et al.~\cite{BalcerCJM20} and Cheu and Ullman~\cite{CheuUllman20}, mentioned in Section~\ref{sec:otherwork}.}

For the common element problem, we show a relationship between the message complexity $\ell$, the input domain size  $|\cal X|$, and the privacy parameters $\epsilon$ and $\delta$. In particular, for constant $\epsilon$ and negligible $\delta$, our bound yields that for constant number of messages $\ell$ and domain size $|\calX|> 2^{n^{O(\ell)}}$ the common element problem does not admit a one-round shuffle model protocol.
At the heart of the lowerbound proof is a transformation from a shuffle model protocol into a local differential privacy randomizer, for which bounds on the mutual information between the input and output are known (see, e.g.,~\cite{Vadhan2017}).

The one-round lowerbound is contrasted in \cref{sec:commonTwoRound} with a two-round protocol for the common element problem where each party sends a {\em single} message in each round. In this protocol, the parties need to communicate through the shuffle in only one of the rounds (and can either use the shuffle or a public channel in the other round).

\smallskip \noindent {\bf An impossibility result for the nested common element problem.} In \cref{sec:nestedImpossibility} we show (for large enough $\calX$, i.e., $|\calX|=\tilde{\Omega}(n^2)$) that, regardless of the number of messages sent by each party, no one-round shuffle protocol exists for the problem that is secure against coalitions of $\alpha n$ semi-honest parties, even when the domain $\calY$ is binary.
We observe that for every $c<1$ the nested common element problem has a 2-round private protocol secure against a coalition of size $cn$.
This gives a separation between what can be computed with coalitions of size up to $\alpha n$ in one- and two-round shuffle model protocols.
Intuitively, the lowerbound follows from the fact that after seeing the shuffle outcome, a coalition covering $P_1,\ldots,P_{\lfloor\alpha n\rfloor}$ can simulate the protocol's execution for any possible value $x\in\calX$ and hence learn all vector entries on which the inputs of parties $P_{\lfloor\alpha n\rfloor+1},\ldots,P_n$ agree. When $\calY$ is binary, Bun et al.~\cite{BunUV18} have used fingerprinting codes to show that this task is impossible when the dimension of the vectors  is $\tilde\Omega(n^2)$, even in the curator model of differential privacy
(in the setting of the nested common element the dimension corresponds to $|\calX|$).\footnote{Bun et al.~\cite{BunUV18} have considered a related problem, however their technique applies also to this task.}

\smallskip \noindent {\bf A one-round protocol for the nested common element problem.} 
A natural approach to solve the nested common element problem in two rounds is to execute a (one-round) protocol for the common element problem among parties $P_1,\ldots,P_{\lfloor\alpha n\rfloor}$, then, if a common element $x$ is found, repeat the protocol with parties
$P_{\lfloor\alpha n\rfloor+1},\ldots,P_n$ ignoring all but the $x$-th entry of their vectors.
It may seem that any shuffle model protocol for the problem should require more than one round. 
In Section~\ref{sec:nestedPossibility} we show that this is not the case. In fact, there is a one-round protocol that tightly matches the above impossibility result for $\alpha \leq 1/2$. For all $c<\min\set{\alpha,1-\alpha}$ there exist one-round shuffle model protocols for the nested common element problem that are secure in the presence of coalitions of size up to $cn$.

\subsection{Other related work}
\label{sec:otherwork}

Private protocols for the common element problem in the shuffle model are implied by protocols for histograms~\cite{CheuSUZZ19,GhaziGKPV19,BC20}. 
Specifically, for all $c<1$, one-round shuffle model protocols for the common element problem that are secure in the presence of coalitions of size up to $cn$ (provided that $n=\Omega(\frac{1}{\epsilon^2} \log \frac{1}{\delta})$) are implied by the protocols of  Balcer and Cheu~\cite{BC20}. While they only considered privacy given the view of the analyzer, their protocols are secure against coalitions containing a constant fraction of the parties. 

Lowerbounds on the error level achievable in the one-round single message shuffle model for the problems of frequency estimation and selection were provided by Ghazi et al.~\cite{GhaziGKPV19}.
Robustness against adversarial behaviour in the shuffle model was informally discussed by Balle et al.~\cite{BalleBGN20}, when discussing the effect  malicious parties can have on the accuracy guarantees in their protocols for addition of real numbers.

Closest to our interest are the recent lowerbounds by Balcer et al.~\cite{BalcerCJM20}. They define robustly shuffle private one-round protocols, where privacy guarantees are required to hold if at least $\gamma n$ parties participate in the protocol. The other {\em malicious} parties avoid sending messages to the shuffle. 
While this model is equivalent to ours in the one-round setting, the lowerbound techniques in~\cite{BalcerCJM20} are different from ours. In particular, they forge an interesting relationships between online pan-privacy~\cite{DworkNPRY10} and robustly shuffle private one-round protocols and hence can use lowerbounds from pan-privacy to deduce lowerbounds for robustly shuffle private one-round protocols. Specifically, for estimating the number of distinct elements they prove that the additive error grows as $\Theta_\epsilon(\sqrt k)$, and for uniformity testing they prove that the sample complexity grows as $\tilde\Theta_{\epsilon,\delta}( k^{2/3})$. In both cases $k$ is the domain size. (These bounds also hold in our model.) As with our bounds, the lowerbounds by Balcer et al.\ hold in the case where different parties may use different randomizers, and send multiple messages.

Independent and parallel to our work, Cheu and Ullman~\cite{CheuUllman20} presented exponential separations between the 1-round shuffle model and the (centralized) curator model of differential privacy. In particular, they showed that every 1-round shuffle model protocol for private agnostic learning of parity functions over $d$ bits requires $\Omega(2^{d/2})$ samples, while $O(d)$ samples suffice in the curator model.
Our work shows, in particular, that private agnostic learning of parity functions using $O(d)$ samples can be done in the shuffle model in two rounds (with semi-honest parties assuming an honest majority). Hence, combined with our work, the results of~\cite{CheuUllman20} provide additional separations between one-round and two-round shuffle model protocols.

\section{Preliminaries}

\subsection{The communication model}

Let $\calX$ be a data domain and let $\calM$ be an arbitrary message domain (w.l.o.g., $\bot \in \calX, \calM$).
We consider a model where the inputs and the computation are distributed among $n$ parties $P_1,\ldots,P_n$ executing a protocol $\Pi=(\bar R, S)$, where $\bar R = (R_1,\ldots,R_n)$ are $n$ stateful randomized functionalities and $S$ is a stateless channel that acts either as a shuffle functionality or as a public channel. See \cref{prot:model} for a formal description of protocols in the shuffle model.

\protocol{Execution of a protocol $\Pi=\left((R_1,\ldots,R_n), S\right)$ in the shuffle model}{The communication model.}{prot:model}{

\textbf{Initialization:} 

\begin{itemize}
\item All parties receive a public random string $\crs\in\{0,1\}^*$.

\item Each party $P_i$ receives its input $x_i \in \calX$  and initializes the execution of $R_i(\crs, x_i)$.
\end{itemize}

\textbf{Communication rounds $1\leq j \leq r$:}

\begin{enumerate}

\item If round $j$ uses $S$ as a shuffle:
\begin{enumerate}
    \item Each party $P_i$ invokes $R_i$ to generate $\ell$ messages $(m_{i,j}[1],\ldots,m_{i,j}[\ell])\in \calM^\ell$ and sends $(m_{i,j}[1],\ldots,m_{i,j}[\ell])\in \calM^\ell$ to $S$. 
    \item Let $(\hat m_1,\ldots,\hat m_{n\ell}) = (m_{1,j}[1],\ldots,m_{1,j}[\ell],\ldots,m_{n,j}[1],\ldots,m_{n,j}[\ell])$ be the $n\ell$ messages received by $S$.
    
    \item $S$ chooses a permutation $\pi:[n\ell]\rightarrow[n\ell]$ uniformly at random.
    
    \item \label{step:shuffleOutput} $S$ outputs $s_j=(\hat m_{\pi(1)},\ldots,\hat m_{\pi(n\ell)})$ to all parties.
\end{enumerate}

\item Otherwise (round $j$ uses $S$ as a public channel):
\begin{enumerate}
    \item Each party $P_i$ invokes $R_i$ to generate a (single) message $m_{i,j}\in \calM$, which it sends to $S$.
    \item $S$ outputs $s_j = (m_{1,j},\ldots,m_{n,j})$.
\end{enumerate}

\item $P_i$ feeds $s_j$ to $R_i$.
\end{enumerate}
\textbf{Output:} Each party $P_i$ invokes $R_i$ to obtain its local output $o_i$.
}

\begin{definition}
Consider an execution of a protocol in the shuffle model as described in \cref{prot:model}.
The {\em message complexity} of \/ $\Pi$ is $\ell$, the number of messages that each party sends to the shuffle in each round. The {\em round complexity} of \/ $\Pi$ is $r$. The {\em shuffle complexity} of \/ $\Pi$ is the number of rounds where $S$ is used as a shuffle.
\end{definition}

\begin{remark}
A protocol that uses a public random string $w$ can always be converted into a protocol that does not use a public random string, at the cost of one additional communication round in which party $P_1$ sends the string $w$ (in the semi-honest setting). This additional communication round can be thought of as an ``offline'' round, as it is independent of the inputs and the function.
\end{remark}

\subsection{Differentially private shuffle model protocols}
\begin{definition}
We say that input vectors $\vec{x} = (x_1,\ldots,x_n)\in \calX^n$ and $\vec{x'}=(x'_1,\ldots,x'_n)\in \calX^n$  are {\em $i$-neighboring} if they differ on exactly the $i$-th entry. We say that $\vec{x}$ and $\vec{x'}$ are {\em neighboring} if there exists an index $i$ such that they are  $i$-neighboring.
\end{definition}

\begin{definition}
We say that two probability distributions ${\cal D}_0, {\cal D}_1\in \Delta(\Omega)$ are $(\epsilon,\delta)$-close and write ${\cal D}_0 \approx_{\epsilon,\delta} {\cal D}_1$ if for all events $T\subset \Omega$ and for $b\in\set{0,1}$, $$\Pr_{t\sim {\cal D}_b}\left[t\in T\right] \leq e^{\epsilon} \cdot \Pr_{t\sim {\cal D}_{1-b}}\left[t\in T\right] + \delta.$$ 
\end{definition}

\begin{definition}[Differential privacy~\cite{DMNS06,DKMMN06}]
An algorithm $\calA$ is {\em $(\epsilon,\delta)$-differentially private} if for all neighboring $\vec{x}, \vec{x}'$ we have that
$
\calA(\vec{x}) \approx_{\epsilon,\delta} \calA(\vec{x}'). 
$
\end{definition}

We are now ready to define what it means for a protocol to be differentially private in the (semi-honest) shuffle model.
Intuitively, this means that the view of every coalition $\calC$ of up to $t$ parties  cannot depend too strongly on the input of a party $P_i\not\in \calC$. 
More formally,

\begin {definition}[View in shuffle model] The view  of a coalition $\calC$ on input $\vec{x}$ in protocol $\Pi$, denoted $\View^\Pi_\calC(\vec{x})$, is the random variable consisting of the public randomness $\crs$, the  inputs and local randomness of the parties in $\calC$, and the output of the $r$ rounds of $\Pi$ when executed on $\vec{x}$, i.e., $s_1,\ldots,s_r$.
\end{definition}

\begin{definition}[Multiparty semi-honest differential privacy \cite{BeimelNO08,Vadhan2017}]
\label{def:DP_Prot}
A protocol $\Pi$ is $(\epsilon,\delta)$-differentially private against coalitions of size $t$ if for all $i\in [n]$, for all coalitions $\calC$ of $t$ parties s.t.\ $P_i\not\in\calC$,  and for all $i$-neighboring $\vec{x}, \vec{x}'$,
$$\View^\Pi_\calC(\vec{x}) \approx_{\epsilon,\delta} \View^\Pi_\calC(\vec{x}').$$
\end{definition}

Observe that if a protocol is differentially private against coalitions of size $t$ as in the definition above, then it also the case that  $\View^\Pi_\calC(\vec{x}) \approx_{\epsilon,\delta} \View^\Pi_\calC(\vec{x}')$ for all coalitions $\calC$ of size less than $t$.

\begin{remark}\ 
\label{rem:model}
\begin{enumerate}
    \item {\bf The shuffle functionality $\boldsymbol{S}$.} It is not essential that the shuffle functionality $S$ be randomized. The shuffle output $s$ in \stepref{step:shuffleOutput} of Protocol $\Pi$ in \cref{prot:model} can be replaced with any canonical representation of the multiset $\{\hat m_1,\ldots,\hat m_{n\ell}\}$ (e.g., in lexicographic order) without affecting any of our results. 

    \item {\bf Hybrid-shuffle model.} The shuffle model can equivalently be thought of as a hybrid model, where all parties have access to a shuffle functionality. 
    
    \item {\bf The local randomizers $\boldsymbol{R_i}$.} In deviation from most of prior work on the shuffle model, the randomizers $R_1,\ldots,R_n$ need not be identical. In particular, the execution of $R_i$ may depend on the identity $i$ of player $P_i$.

    \item {\bf Local model protocols.} An $(\epsilon,\delta)$-differentially private protocol $\Pi$ with zero shuffle complexity  satisfies local differential privacy~\cite{KLNRS11,Vadhan2017}.

    \item {\bf Shuffle model with an analyzer.} In prior work on the shuffle model one party, $A$, is an {\em analyzer}. The analyzer has no input ($x_A=\bot$) and does not send messages, i.e., ($m_{A,j}[1], \ldots, m_{A,j}[\ell]) = \bot^\ell$ for $1\leq j \leq r$. In this setting the local output of parties $P_1,\ldots,P_n$ is $\bot$ and the outcome of the protocol is the local output of $A$. Sections~\ref{sec:commonElementProblem} and~\ref{sec:possibilityImpossibility} consider the shuffle model with an analyzer.
\end{enumerate}
\end{remark}

\subsection{Secure computation protocols with semi-honest parties}

Let $f:\calX^n\rightarrow\calY^n$ be a randomized functionality. We recall the definition from the  cryptographic literature of what it means that a protocol $\Pi$ securely computes $f(x_1,\ldots,x_n)$ with semi-honest parties. 
We will use this definition both in the shuffle model and in the setting where the parties communicate over a complete network of private channels. 
For the latter we define the view of a coalition as follows:
\begin{definition}[View in complete network of private channels] The view of a coalition 
$\calC$ on input $\vec{x}$ in protocol $\Pi$, denoted $\view^\pi_\calC(\vec{x})$, is the random variable consisting of the inputs and local randomness of the parties in $\calC$ and the messages the parties in $\calC$ receive from the parties in $\overline{\calC}= \{P_1,\ldots,P_n\}\setminus \calC$. \end{definition}

\begin{definition}[Secure computation in the semi-honest model] 
\label{def:secureComputation}
A protocol \/ $\Pi$ is said to  $\delta$-securely compute $f$ with coalitions of size at most $t$ if there exists a simulator $\Sim^\Pi$ such that for any coalition $\calC$ of at most $t$ parties and every input vector $\vec{x} = (x_1,\ldots,x_n) \in \calX^n$,
$$\left(\Sim^\Pi(\calC, \vec{x}[\calC],\vec{y}[\calC]), \vec{y}[\overline\calC]\right)\approx_{0,\delta} \left(\View^\Pi_\calC(\vec{x}), \Output(\overline\calC)\right),$$
where $\vec{y}=f(\vec{x})$ and $\Output(\overline\calC)$ is the output of the parties in $\overline\calC$
in the protocol. The probability distribution on the left is over the randomness of $f$ and  the randomness of the simulator, and the probability distribution on the right is over the randomness of the honest parties and the adversary. When $\delta=0$ we say that $\Pi$ provides perfect privacy.
\end{definition}

\begin{remark}
In the shuffle model, $\View^\Pi_\calC(\vec{x})$ also includes  the public random string $w$ (if exists), and the probability distribution on the right in \cref{def:secureComputation} is also over the public random string.
\end{remark}

We next state a composition theorem for differentially private protocols using secure  protocols.

\begin{lemma}
\label{lem:composition}
Let  \/  $\Pi$ be a protocol with one invocation of a black-box access to some function $f$ (the $f$-hybrid model). 
Let  \/  $\Pi_f$ be a protocol 
that $\delta'$-securely computes $f$ with coalitions of size up to $t$. 
Let \/  $\Pi'$ be as in $\Pi$, except that the call to $f$ is replaced with the execution of \/ $\Pi_f$. 
If \/ $\Pi$ is $(\epsilon,\delta)$-differentially private with coalitions of size up to $t$, then $\Pi'$ is $(\epsilon,(e^\epsilon+1)\cdot\delta' + \delta)$-differentially private with coalitions of size up to $t$.
\end{lemma}

\begin{proof}
Consider a coalition $\calC$ of up to $t$ parties.
The random variable $\View_\calC^{\Pi'}(\vec{x})$ consisting the view of coalition $\calC$ in an execution of protocol $\Pi'$ can be parsed into the view of $\calC$ in protocol $\Pi$, i.e., $\View_\calC^{\Pi}(\vec{x})$, and the view of $\calC$ in the execution of protocol $\Pi_f$, i.e.,  $\View_\calC^{\Pi_f}(\vec{y})$. In the latter $\vec{y}$ is the input to $f$ in the execution of $\Pi$ on input $\vec{x}$ (similarly, we will use $\vec{y}'$ to denote the  input to $f$ in the execution of $\Pi$ on input $\vec{x}'$). 
Note that, by \cref{def:secureComputation}, $\View_\calC^{\Pi_f}(\vec{y})$ can be simulated as $\Sim^{\Pi_f}(\calC, \vec{y}[\calC],f_\calC(\vec{y}))$ up to statistical distance $\delta'$. 
Observe that  $\View_\calC^\Pi$ contains the inputs $\vec{y}_\calC$ sent to $f$ as well as the outcome seen by the coalition, $f_\calC(\vec{y})$. Hence, $\Sim^{\Pi_f}(\calC, \vec{y}[\calC],f_\calC(\vec{y}))$ is a post-processing of $\View_\calC^{\Pi}(\vec{x})$. To emphasize this fact, we write $\Sim^{\Pi_f}(\View_\calC^{\Pi}(\vec{x}))$ instead of $\Sim^{\Pi_f}(\calC, \vec{y}[\calC],f_\calC(\vec{y}))$.

Let $P_i\not\in\calC$. For all $i$-neighboring $\vec{x},\vec{x'}$ and all $T$ we have that 
\begin{align*}
\Pr[\View_\calC^{\Pi'}(\vec{x}) \in T] 
& =
\Pr[(\View_\calC^{\Pi}(\vec{x}), \View_\calC^{\Pi_f}(\vec{y})) \in T] \\
& \leq  
\Pr[(\View_\calC^{\Pi}(\vec{x}), \Sim^{\Pi_f}(\View_\calC^{\Pi}(\vec{x})))\in T] + \delta' & \mbox{(\cref{def:secureComputation})}\\
& \leq 
e^\epsilon\cdot\Pr[(\View_\calC^{\Pi}(\vec{x'}), \Sim^{\Pi_f}(\View_\calC^{\Pi}(\vec{x'})))\in T] + \delta + \delta' & \mbox{$\Pi$ is $(\epsilon,\delta)$-DP)} \\
& \leq
e^\epsilon\cdot(\Pr[(\View_\calC^{\Pi}(\vec{x}'), \View_\calC^{\Pi_f}(\vec{y}')) \in T]+\delta') + \delta + 
\delta' & \mbox{(\cref{def:secureComputation})}\\
& = e^\epsilon\cdot\Pr[\View_\calC^{\Pi'}(\vec{x}') \in T]+(e^\epsilon + 1) \delta' + \delta.
\end{align*}

The second step in the analysis follows from the fact that differential privacy is preserved under post-processing.
\end{proof}

\subsection{Pairwise independent hash functions}

In our constructions We use pair pairwise independent hash functions, defined below.
\begin{definition}[Pairwise independent hash functions]
A family $H = \set{h : \calX \rightarrow R}$ is
said to be pairwise independent, if for any two distinct elements $x_1 \neq x_2 \in \calX$, and any two
(possibly equal) values $y_1, y_2 \in R$,
$$\Pr_{h\in H} [h(x_1) = y_1 \wedge h(x_2) = y_2] = \frac{1}{|R|^2},$$
where $h$ is chosen with uniform distribution from $H$ independently of $x_1,x_2$.
\end{definition}
In particular, if $H$ is a pairwise independent family, then for every $x_1 \neq x_2 \in \calX$ it holds that
$\Pr_{h\in H} [h(x_1) = h(x_2) ] = \frac{1}{|R|}$, and for every set $A \subseteq \calX$ we have
$\Pr_{h\in H} [\exists_{x_1 \neq x_2 \in A}\; h(x_1) = h(x_2) ] \leq \frac{|A|^2}{|R|}$,
in this case we say that $A$ is perfectly hashed by $h$.

\section{A Two-Round Secure MPC Protocol in the Shuffle Model}

In this section we show that every functionality  that can be computed with differential privacy in the centralized model can be  computed with differential privacy 
in the shuffle model 
in two rounds assuming an honest majority. To achieve this result we first show a one-round protocol in the shuffle model for secure message transmission, that is, we show that how to emulate a private channel. This result together with an honest-majority two-round MPC protocol of~\cite{ApplebaumBT18} in the private channel model imply that every functionality (including differentially-private functionalities) can be securely computed 
in the shuffle model in two rounds assuming an honest majority.

\subsection{A one-round secure message transmission protocol}
\label{sec:secureMessageTransmission}

Assume that party $P_i$ wants to send a message to party $P_j$ using the shuffle such that any other party will not learn any information on the message. In~\cite{IKOS06} this was done in two rounds. In the first round $P_i$ and $P_j$ agree on a secret key, and in the second round $P_i$ encrypts the message using this key  as a one-time pad. We present a  protocol such that $P_i$ knows the key in advance and can encrypt the message already in the first round. The resulting protocol has statistical security. 

We start by describing a variant of the protocol of~\cite{IKOS06} for key exchange. As a first step, we describe a key exchange protocol in which $P_i$ and $P_j$ agree with probability $1/2$ on a random bit (and with probability $1/2$ the output is ``FAIL''). The protocol is as follows: 
Party $P_i$ samples a uniformly distributed bit $a$ and sends to the shuffle the message $(i,j,a)$.
Similarly, party $P_j$ samples a uniformly distributed bit $b$ and sends to the shuffle the message $(i,j,b)$.\footnote{
We add the prefix $i,j$ to the messages sent by $P_i$ and $P_j$ to enable all pairs of parties to exchange keys in parallel. It is essential that both $P_i$ and $P_j$ list the identities $i,j$ in the same order (e.g., lexicographic order).}
If $a = b$ the protocol fails. Otherwise, the joint key is $a$. As both parties $P_i,P_j$ get the output of the shuffle, they both know if the protocol fails ($a=b$) or not, and if the protocol does not fail ($a\neq b$) they both know $a$ -- the common key. On the other hand, an adversary that sees the output of the shuffle when $a\neq b$, sees a shuffle of the two messages $\{(i,j,0),(i,j,1)\}$ and does not get any information on $a$. 
To generate a $k$-bit key, the above protocol is repeated  $3k$ times in parallel with independent random bits $a_\ell,b_\ell$  in each execution, and the shared key is the bits of $P_i$ in the first $k$ indices where $a_\ell \neq b_\ell$. 
By a simple Chernoff-Hoefding bound, the probability that there are no such $k$ indices is exponentially small.
See \cref{prot:KeyExchange} for a formal description of the protocol.   

\protocol{\ProtKeyExchange}{A one-round key exchange protocol.}{prot:KeyExchange}{
\textbf{Inputs:} $P_i$ and  $P_j$ hold a security parameter $1^k$.
\begin{enumerate}
    \item
    $P_i$ samples $3k$ uniformly distributed bits $(a_1,\dots,a_{3k})$ and sends to the shuffle $3k$ messages $(i,j,1,a_1),\dots,(i,j,3k,a_{3k})$.
    \item
    $P_j$ samples $3k$ uniformly distributed bits $(b_1,\dots,b_{3k})$ and sends to the shuffle $3k$ messages $(i,j,1,b_1),\dots,(i,j,3k,b_{3k})$.
    \item
    The shuffle publishes a random permutation of the messages it got.
    \item
    Let $\ell_1 < \ell_2 < \cdots < \ell_k$ be the first $k$ indices such that $a_{\ell_j}\neq b_{\ell_j}$
    (if there are no such $k$ indices, output ``FAIL''). The joint key is $(a_{\ell_1},a_{\ell_2},\ldots,a_{\ell_k})$.
\end{enumerate}
}

To construct a one-round protocol for secure message transmission from $P_i$ to $P_j$, we want $P_i$ to know the key in advance so it can use the key to encrypt the message at the same time it sends the messages for the key exchange.  
In \ProtKeyExchange, party $P_i$ does not know the key in advance since it does not know the bits that $(a_1,\ldots,a_{3k})$ and $(b_1,\dots,b_{3k})$ disagree. To overcome this problem $P_i$ will use all the bits it generates as a pre-key $K$. In this case $P_j$ will know all bits of the pre-key $K$ whereas an adversary will learn only about half of the bits of $K$. Parties $P_i$ and $P_j$ wish to agree on a key generated from the pre-key $K$ without interaction such that the adversary gets  negligible information about the agreed key. This is an instance of the privacy amplification problem and a simple solution is to sample a pairwise independent hash function $h$ and set the key as $h(K)$. It follows by the left-over hash lemma~\cite{HastadILL99} that $h(K)$ is close to uniform given $h$ and the knowledge of the adversary about the pre-key $K$.

\begin{theorem}[The left-over hash lemma~\cite{HastadILL99}]
\label{thm:leftover}
Let $m,n$ be integers and $X$  be a random variable  distributed over $\set{0,1}^n$ 
such that $\Pr[X=x]\leq 2^{-m}$ for every $x \in \set{0,1}^n$.
Let $\calH$ be a family of pairwise independent hash functions from $\set{0,1}^n$ to $\set{0,1}^{m-2k}$. 
Then, for a random $h$ uniformly distributed in $\calH$  and independent of $X$, 
$$\SD \left((h(X),h),(U,h)\right)\leq 2^{-k},$$
where $U$ is uniform over $\set{0,1}^{m-2k}$ and independent of $h$, and where $\SD$ denotes the statistical distance (total variation distance).
\end{theorem}

\protocol{\ProtMessageTransmission}{A one-round protocol for secure message transmission.}{prot:MessageTransmission}{
\textbf{Inputs:} Party $P_i$ holds a security parameter $1^k$ and a message $M$ of length at most $k$, party $P_j$ holds  security parameter $1^k$.
\begin{enumerate}
    \item
    $P_i$ samples $7k$ uniformly distributed bits $(a_1,\dots,a_{7k})$ and sends to the shuffle $7k$ messages $(i,j,1,a_1),\dots,(i,j,7k,a_{7k})$.
    \item
    $P_j$ samples $7k$ uniformly distributed bits $(b_1,\dots,b_{7k})$ and sends to the shuffle $7k$ messages $(i,j,1,b_1),\dots,(i,j,7k,b_{7k})$.
    \item
    \label{step:encrypt}
    $P_i$ samples a function $h$ uniformly at random from a family of pairwise independent functions 
    $\calH=\set{h:\set{0,1}^{7k}\rightarrow \set{0,1}^{k}}$ and sends  to the shuffle the message $(i,j,\text{\rm ``message''},h,h(a_1,\ldots,a_{7k})\oplus M)$.
    \item
    The shuffle publishes a random permutation of the messages it got.
\end{enumerate}
}
\begin{theorem}
\label{thm:SMT}
\ProtMessageTransmission is a correct and secure protocol for message transmission, that is (1) $P_j$ can always recover $M$, (2) For every two messages $M,M'$ the statistical distance between the views of the referee and all parties except for $P_i$ and $P_j$ in an executions of \ProtMessageTransmission with $M$ and 
\ProtMessageTransmission with $M'$ is at most $3\cdot 2^{-k}$. 
\end{theorem}
\begin{proof}
For the correctness of the protocol, as $P_j$ knows its messages, it can deduce for every $\ell$ the message $(i,j,\ell,a_\ell)$ sent by $P_i$, hence compute the common key $h(a_1,\ldots,a_{7k})$ and compute $M$. 

For the security of the protocol, first note that by a Chernoff-Hoefding bound, the probability that there are less than $3k$ indices $\ell$ such that $a_\ell\neq b_\ell$ is less than $2^{-k}$, and such executions add at most $2^{-k}$ to the statistical distance. We continue the analysis  assuming that such event did not occur. 

 We consider an execution of \ProtMessageTransmission in which is \stepref{step:encrypt} party $P_i$ sends the message $(i,j,\text{\rm ``message''},h,u\oplus M)$ for a uniformly sampled $u \in \set{0,1}^k$. In this case, the executions for $M$ and $M'$ are equally distributed (as $u$ acts as a one-time pad). To prove the security it suffices to prove that for every message $M$, the statistical distance in the view in the executions of \ProtMessageTransmission and 
 the modified \ProtMessageTransmission (both with $M$) is at most $2^{-k}$.
Fix a set $L \subset [7k]$ of size at least $3k$, and consider all executions in which $a_\ell \neq b_\ell$ if and only if $\ell \in L$.
For every index $\ell\in L$, the view discloses no information on $a_\ell$ in these executions (since an adversary sees a random shuffle of the two messages $(i,j,\ell,0),(i,j,\ell,1)$ and does not get any information on $a_\ell$).    
In other words, there are at least $2^{3k}$ strings $(a_1,\dots,a_{7k})$ possible given the executions are consistent with $L$, and all strings are equiprobable. Thus, by \cref{thm:leftover}, the statistical distance between $u$ and $h(a_1,\dots,a_{7k})$ is at most $2^{-k}$. This completes the proof of security.
\end{proof}

\subsection{A two round MPC protocol}
\label{sec:twoRoundMPC}

We construct a two-round MPC protocol in the shuffle model for every functionality on inputs from a finite domain assuming an honest majority. The construction is via a combination of the two-round MPC protocol of Applebaum, Brakersky, and Tsabary~\cite{ApplebaumBT18} (Henceforth, \ProtABT, see \cref{thm:ABT} below), which assumes private channels between every pair of parties, with \ProtMessageTransmission executed in the shuffle model. The latter is used for simulating the private channels. 
\begin{theorem}[\ProtABT \protect{\cite[Theorem 1.1]{ApplebaumBT18}}]
\label{thm:ABT}
At the presence of honest majority, any function $f$ can be computed with perfect privacy
in a complete network of private channels
in two rounds with polynomial efficiency in the number of parties and in the size of the formula that computes $f$.
\end{theorem}

\begin{theorem}
\label{thm:2roundMPC}
Let $f:\calX^n\rightarrow \set{0,1}$ be a function and $\gamma>0$ ($\gamma$ can depend on $n$ and $f$).
At the presence of honest majority, any function $f$ can be computed with $\gamma$-statistical privacy
in the shuffle model 
in two rounds with polynomial efficiency in the number of parties, in the size of the formula that computes $f$, and in $\log 1/\gamma$.
\end{theorem}

\begin{proof}
In \cref{prot:MPCinShuffle}, we describe  \ProtMPCinShuffle\ -- the two round MPC protocol in the shuffle model.

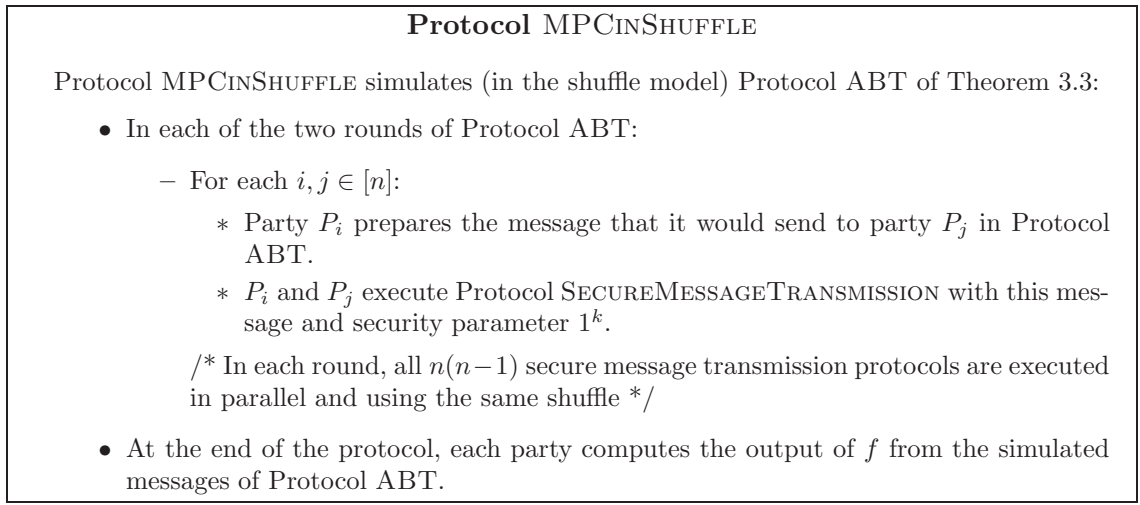
\protocol{\ProtMPCinShuffle}{A two-round MPC protocol in the shuffle model for arbitrary functionalities.}{prot:MPCinShuffle}{
 \ProtMPCinShuffle simulates (in the shuffle model)  \ProtABT of~\cref{thm:ABT}: 
 \begin{itemize}
     \item In each of the two rounds of  \ProtABT:
     \begin{itemize}
         \item For each $i,j\in [n]$:
         \begin{itemize}
             \item Party $P_i$ prepares the message that it would send to party $P_j$ in  \ProtABT.
             \item $P_i$ and $P_j$ execute \ProtMessageTransmission with this message and security parameter $1^k$.
         \end{itemize}
          \item[]   /* In each round, all $n(n-1)$ secure message transmission protocols  are executed in parallel and using the same shuffle */
     \end{itemize}
     \item At the end of the protocol, each party computes the output of $f$ from the simulated messages of \ProtABT.
\end{itemize}  
}

As \ProtMessageTransmission has perfect correctness,  each party in \ProtMPCinShuffle
can compute the messages it gets in  \ProtABT and compute $f$ without any error. 

For the security of the protocol, let $\calC$ be a coalition of less than $n/2$ parties. We construct a simulator that generates a view for $\calC$ that is $O(n^2 2^{-k})$ far from the view of $\calC$ in the real-world execution of \ProtMPCinShuffle: 
\begin{itemize}
    \item  Execute the simulator of  \ProtABT of \cref{thm:ABT} and generate a view for $\calC$ that is identically distributed as the real view of  $\calC$ in  \ProtABT. 
    \item
    For each round and for each pair $P_i,P_j$: 
\begin{itemize}
\item 
If at least one of $P_i,P_j$ is in $\calC$ 
then let $M_{i,j}$ be the message that $P_i$ sends to $P_j$ in the simulated view.
\item
Otherwise, let $M_{i,j}$ be some fixed arbitrary message.
\item
Execute \ProtMessageTransmission with the message $M_{i,j}$ and generate the messages that $P_i,P_j$ send to the shuffle.
\end{itemize}
\item
For each round, shuffle the messages generated by $P_i,P_j$ for every $i,j\in [n]$. 
\item \textbf{Output:} The shuffled messages of round 1 and  the shuffled messages of round 2,
the randomness of every $P_i$ generated by the simulator of  \ProtABT, and the randomness used by every $P_i \in \calC$ in an execution of
\ProtMessageTransmission for which $P_i$ is either the sender or the receiver.
\end{itemize}

By \cref{thm:SMT}, for every $P_i,P_j\notin \calC$, the messages generated in the simulation (i.e., the messages of \ProtMessageTransmission for the fixed message $M_{i,j}$ and the message that $P_i$ and $P_j$ send to the shuffle in the real world for the real message of the \ProtABT of \cref{thm:ABT} are only $O(2^{-k})$ far. Thus,
the output of the simulator we constructed is at most $O(n2^{-k})$ far from the view of $\calC$ in the real execution of \ProtMPCinShuffle.
\end{proof}

\begin{remark} \ 
\begin{enumerate}
\item In \ProtMessageTransmission we use the shuffle in both rounds as we execute 
\ProtMessageTransmission in each round. We can optimize the protocol and only use the shuffle in the first round. To achieve this, in the first round each ordered pair of parties $P_i,P_j$ also executes \ProtKeyExchange in round 1 and generate a key, which is used by $P_i$ to encrypt the message that it send to $P_j$ in round 2. The encrypted messages is sent on the public channel. 

\item In a setting with an analyzer as in \cref{rem:model}, the protocol can be simplified, with the expense that we now need to assume that the number of colluding parties in $P_1,\ldots,P_n$ is less than $(n-1)/2$. We execute \ProtABT with $n+1$ parties, where the $(n+1)$-th party (i.e., the analyzer) has no input and is the only party that receives an output. Furthermore, we assume that the analyzer is always in the coalition, and, therefore, the messages that it sends and receives are public. As the analyzer cannot send messages to the shuffle, we use the public random string as the random string of the analyzer and the messages that the input-less analyzer sends in the first round to party $P_j$ in \ProtABT are generated by $P_j$ without interaction using the random common string. Furthermore, in the second round each party only sends its message to the analyzer and this message is sent in the clear.
\item
In \ProtMessageTransmission the shuffle receives $O(k)$ messages and shuffles them. We actually only need to shuffle every pair of messages 
$(i,j,\ell,a_\ell),(i,j,\ell,b_\ell)$, thus, we can use many copies of 2-message shuffle. The same is true for \ProtMPCinShuffle.
\end{enumerate}
\end{remark}
\begin{corollary}
Let $f$ be an $(\epsilon,\delta)$-differentially private functional (in the centralized model) acting on inputs from a finite domain and using a finite number of random bits  and $\gamma>0$. At the presence of honest majority, the functionality $f$ can be computed with $(\epsilon,\delta+(e^{\epsilon}+1)\gamma)$-differential privacy
in the shuffle model 
in two rounds with polynomial efficiency in the number of parties, in the size of the formula that computes $f$, and in $\log1/\gamma$.
\end{corollary} 
\begin{proof}
We use \ProtMPCinShuffle to compute the function $f$. By \cref{lem:composition} the resulting protocol is private.
\end{proof}

\section{The Common Element Problem}
\label{sec:commonElementProblem}

In this section we study the following problem.

\begin{definition}[The common element problem]
\label{def:CommonElementProblem}
In the {\em common element} problem, there are $n$ parties $P_1,\dots,P_n$, where each party $P_i$ gets an input $x_i\in\calX$, and there is an analyzer $P_0$ (with no input). If all inputs are equal, i.e., $x_1=x_2=\dots=x_n$, then with probability at least $3/4$ the analyzer must output $x_1$ at the end of the execution. The outcome is not restricted otherwise.
\end{definition}

\subsection{An impossibility result for single-round constant-message protocols}
\label{sec:lowerbound}

We present an impossibility result for 1-round protocols for the common element problem.
Informally, we show that if the domain size $|\calX|$ is large, then either the number of messages $\ell$ must be large, or else the privacy parameter $\delta$ must be ``large''. Before we state and prove this impossibility result, we introduce the following bound on the mutual information between the input of a party in a 1-round differentially protocol and the messages she submits to the shuffle. This bound holds for any 1-round differentially protocol (not only for protocols for the common element problem). The results we use from information theory are given in \cref{app:info}.

\begin{theorem}\label{thm:MutualInfoBound}
Let $\Pi$ be a 1-round shuffle model protocol for $n$ parties satisfying $(\eps,\delta)$-differential privacy for coalitions of size 1, with message complexity $\ell$. Let $\calX$ denote the input domain (i.e., the input of every party is an element of $\calX$). Let $(Z_1,\dots,Z_n)\in\calX^n$ denote (possibly correlated) random variables.
Consider the execution of $\Pi$ on inputs $x_1=Z_1,\dots x_n=Z_n$, and for $i\in[n]$ let $Y_i$ denote the vector of messages submitted by party $P_i$ to the shuffle, in lexicographic order. Also let $\CRS$ be a random variable denoting the public randomness of the protocol. Then for every $i\in[n]$, if $Z_i$ is uniformly distributed over $\calX$ then
$$
I(Y_i,\CRS;Z_i)=O\left( (en)^{\ell}\cdot\left( \epsilon^2 + \frac{\delta}{\epsilon}\log|\calX| + \frac{\delta}{\epsilon}\log\frac{\epsilon}{\delta}  \right) + \ell\cdot \log\left(n\right)\right).
$$
\end{theorem}

In words, the theorem states that the mutual information between $Z_i$ (the input of party $P_i$), and $(Y_i,\CRS)$ (the messages submitted by party $P_i$ and the public randomness) is bounded. 

Before proving \cref{thm:MutualInfoBound}, we quote two basic results from information theory (see \cref{app:info} for the proofs of these lemmas, as well as additional preliminaries form information theory).
Consider three random variables $Y_1,Y_2,Z$, where $Y_1$ and $Y_2$ are conditionally independent given $Z$.
The following lemma shows that the amount of information that $(Y_1,Y_2)$ give about $Z$, is at most the amount that $Y_1$ gives on $Z$ plus the amount that $Y_2$ gives on $Z$. (This is not necessarily true without the conditionally independent assumption.)

\begin{lemma}\label{lem:condIndInformation}
Let $Y_1,Y_2,Z$ be random variables, where $Y_1$ and $Y_2$ are conditionally independent given $Z$. Then, $I(Z;Y_1) + I(Z;Y_2) \geq I(Z; Y_1,Y_2).$
\end{lemma}

The following lemma shows that if $I(X;Y|Z)$ is high and if $H(Z)$ is low, then $I(X;Y)$ must also be high. That is, if $X$ gives a lot of information on $Y$ when conditioning on a random variable $Z$ with low entropy, then $X$ gives a lot of information on $Y$ even without conditioning on $Z$.

\begin{lemma}\label{lem:dontLoseMuchI}
Let $X,Y,Z$ be three random variables. Then, $I(X;Y)\geq I(X;Y|Z)-H(Z).$
\end{lemma}

We are now ready to prove \cref{thm:MutualInfoBound}.
\begin{proof}[Proof of \cref{thm:MutualInfoBound}]
Let $R_1,\dots,R_n$ denote the randomizers in the protocol $\Pi$, and fix $i\in[n]$. 
We use $\Pi$ and $i$ to construct the following algorithm, which we call \texttt{LocalRandomizer}, that gets a single input $x_i$ and a public random string $\crs$.

\begin{enumerate}
    \item Compute $\vec{\widetilde{m}_i}\leftarrow R_i(\crs,x_i)$. That is, $\vec{\widetilde{m}_i}$ is the vector of $\ell$ messages chosen by $R_i$.
    \item For $j\neq i$, sample $x_j\in\calX$ uniformly at random, and let $\vec{\widetilde{m}_j}\leftarrow R_j(\crs,x_j)$.
    \item For $j\in[n]$, we write $\vec{\widetilde{y}_j}$ to denote $\vec{\widetilde{m}_j}$ after sorting it in lexicographic order.
    \item Let $\vec{\widetilde{s}}$ be a random permutation of the collection of all messages in $\vec{\widetilde{m}_1},\dots,\vec{\widetilde{m}_n}$.
    \item Let $\vec{\widetilde{y}}$ denote a (sorted) vector of $\ell$ messages chosen randomly (without repetition) from $\vec{\widetilde{s}}$.
    \item Return $\vec{\widetilde{y}},\crs$.
\end{enumerate}

Consider the execution of \texttt{LocalRandomizer} on a uniformly random input $x_i=\widetilde{Z}$ with the public randomness $\widetilde{\CRS}$. We will use $\widetilde{Y},\widetilde{S}$ and $\left\{\widetilde{M}_i\right\}_{i\in[n]}$ $\left\{\widetilde{Y}_i\right\}_{i\in[n]}$ to denote the random variables taking values $\vec{\widetilde{y}},\vec{\widetilde{s}}, \{\vec{\widetilde{m}_i}\}_{i\in[n]}$, and $\{\vec{\widetilde{y}_i}\}_{i\in[n]}$ during the execution.

Observe that $\widetilde{S}$ is identically distributed to the outcome of the shuffler in an execution of $\Pi$ on random inputs, and observe that the outcome of \texttt{LocalRandomizer} is computed as a post-processing of $\widetilde{S}$ and $\widetilde{\CRS}$. Algorithm \texttt{LocalRandomizer} is, therefore, $(\epsilon,\delta)$-differentially private (as a function of $x_i$). Since the mutual information between the input and the output of a differentially private algorithm is bounded (see, e.g.,~\cite{BassilyS15} or \cref{thm:DPinformation}), there exists a constant $\lambda$ such that
\begin{equation}
I\left(\widetilde{Y},\widetilde{\CRS};\widetilde{Z}\right)\leq \lambda\cdot\left( \epsilon^2 + \frac{\delta}{\epsilon}\log|\calX| + \frac{\delta}{\epsilon}\log(\epsilon/\delta)  \right).\label{eq:lower1}    
\end{equation}
We now relate $I\left(\widetilde{Y},\widetilde{\CRS};\widetilde{Z}\right)$ to $I\left(\widetilde{Y}_i,\widetilde{\CRS};\widetilde{Z}\right)$. Intuitively, the connection is that with probability $\approx n^{-\ell}$ we get that $\widetilde{Y}=\widetilde{Y}_i$. Formally, let $T$ be a random variable taking value 0 if $\widetilde{Y}=\widetilde{Y}_i$ and otherwise $T=1$, and denote $p=\Pr[T=0]=1/\binom{\ell n}{\ell}$. By \cref{lem:dontLoseMuchI} and
using standard bounds on the entropy of a binary random variable (see, e.g.,~\cref{claim:binaryHbounds})  
we get that
\begin{align}
I\left(\widetilde{Y},\widetilde{\CRS};\widetilde{Z}\right)&\geq I\left(\left.\widetilde{Y},\widetilde{\CRS};\widetilde{Z}\right|T\right)-H(T)
\geq I\left(\left.\widetilde{Y},\widetilde{\CRS};\widetilde{Z}\right|T\right)-p\log\left(\frac{4}{p}\right)\nonumber\\
&= \E_{t\leftarrow T}\left[ I\left(\left.\widetilde{Y},\widetilde{\CRS};\widetilde{Z}\right|T=t\right) \right] -p\log\left(\frac{4}{p}\right)
\geq p\cdot I\left(\left.\widetilde{Y},\widetilde{\CRS};\widetilde{Z}\right|T=0\right) -p\log\left(\frac{4}{p}\right)\nonumber\\
&= p\cdot I(\widetilde{Y}_i,\widetilde{\CRS};\widetilde{Z}) -p\log\left(\frac{4}{p}\right).\label{eq:lower2}
\end{align}
So, combining Inequalities~(\ref{eq:lower1}) and~(\ref{eq:lower2}) we get that
\begin{align*}
I\left(\widetilde{Y}_i,\widetilde{\CRS};\widetilde{Z}\right)&\leq\frac{\lambda}{p}\cdot\left( \epsilon^2 + \frac{\delta}{\epsilon}\log|\calX| + \frac{\delta}{\epsilon}\log(\epsilon/\delta)  \right) + \log\left(\frac{4}{p}\right)\\
&\leq \lambda\cdot (en)^{\ell}\cdot\left( \epsilon^2 + \frac{\delta}{\epsilon}\log|\calX| + \frac{\delta}{\epsilon}\log(\epsilon/\delta)  \right) + \ell\cdot \log\left(4en\right).
\end{align*}
Finally, observe that
the input $\widetilde{Z}$, the public randomness $\widetilde{\CRS}$, and the (sorted) vectors of messages $\widetilde{Y}_i$ in the execution of \texttt{LocalRandomizer} are identically distributed to these variables in the execution of $\Pi$ on inputs $(Z_1,\dots,Z_n)$ with the public randomness $\CRS$. That is, the random variables 
$\left(\widetilde{Y}_i,\widetilde{\CRS},\widetilde{Z}\right)$ and $\left(Y_i,\CRS,Z_i\right)$ are identically distributed. Therefore,
\begin{align*}
I\left(Y_i,\CRS;Z_i\right)\leq \lambda\cdot (en)^{\ell}\cdot\left( \epsilon^2 + \frac{\delta}{\epsilon}\log|\calX| + \frac{\delta}{\epsilon}\log(\epsilon/\delta)  \right) + \ell\cdot \log\left(4en\right).  
\end{align*}
\end{proof}

We next present our impossibility result for the common element problem. 

\begin{theorem}\label{thm:lowerBound}
There exists a constant $\lambda>1$ such that the following holds.
Let $\epsilon\leq1$, let $\ell\in\N$, and let $\calX$ be such that $|\calX|\geq2^{\lambda(4en)^{\ell+1}}$.  
Let $\Pi$ be a 1-round protocol for the common element problem over the domain $\calX$ with message complexity $\ell$, such that $\Pi$ is $(\epsilon,\delta)$-differentially private for coalitions of size 1. Then,
$$
\delta=\Omega\left( (en)^{-\ell-1}\right).
$$
\end{theorem}

\begin{proof}
We first give a short overview of the proof.  Recall that if all inputs are equal to some element $x\in\calX$, then the analyzer must output $x$ with high probability. This also holds when the (common) input $x$ is chosen uniformly at random from $\calX$, which means that the mutual information between the (common) input and the output of the analyzer must be high. We show that this means that there must be at least one party $P_{i^*}$ such that mutual information between the random (common) input and the messages submitted by $P_{i^*}$ must be high, which will contradict Theorem~\ref{thm:MutualInfoBound}. 

Let $R_1,\dots,R_n$ denote the randomizers in the protocol $\Pi$. Let $Z$ be a uniformly random element of $\calX$ and consider the execution of $\Pi$ on inputs $x_1=x_2=\dots=x_n=Z$ with a public random string $\CRS$. For $i\in[n]$, let $M_i$ denote a random variable representing the vector of $\ell$ messages submitted to the shuffler by party $P_i$, and let $Y_i$ be the same as $M_i$ after sorting it in lexicographic order. Let $S$ be a random variable denoting the outcome of the shuffler. That is, $S$ is a random permutation of all the messages in $M_1,\dots,M_n$.
Alternatively, $S$ is a random permutation of all the messages in $Y_1,\dots,Y_n$.
We use $A$ for the random variable denoting the outcome of the analyzer at the end of the execution.

Since $A=Z$ with probability at least $3/4$,  the mutual information between $A$ and $Z$ must be high. Specifically, Let $B$ be a random variable taking value 0 if $A=Z$ and otherwise $B=1$. By \cref{lem:dontLoseMuchI}
\begin{align*}
I(A;Z)&\geq I(A;Z|B)-H(B)
\geq I(A;Z|B)-1
= \E_{b\leftarrow B}\Big[ I(A;Z|B=b) \Big] - 1\\
&\geq \frac{3}{4}\cdot I(A;Z|B=0) - 1
=  \frac{3}{4}\cdot I(Z;Z) - 1
=  \frac{3}{4}\cdot H(Z) - 1
= \frac{3}{4}\cdot \log|\calX| - 1
\geq\frac{1}{2}\cdot \log|\calX|.
\end{align*}

Recall that $A$ is a (possibly randomized) function of the outcome of the shuffle $S$ and the public randomness $\CRS$. Hence,
$
I(S,\CRS;Z) \geq I(A;Z) \geq\frac{1}{2}\cdot \log|\calX|.
$
We now show that there must exist an index $i^*\in[n]$ such that
$$
I(Y_{i^*},\CRS;Z)\geq\frac{1}{n}\cdot I(S,\CRS;Z)\geq\frac{1}{2n}\cdot \log|\calX|.
$$
To that end, observe that since $\Pi$ is a 1-round  protocol, then conditioned on $Z$ and on the public randomness $\CRS$ we have that the messages that party $P_i$ sends are independent of the messages that party $P_j$, where $j\neq i$, sends. That is, the random variables $Y_1,\dots,Y_n$ are conditionally independent given $(Z,\CRS)$. Therefore, by \cref{lem:condIndInformation} we have that
\begin{align*}
\sum_{i\in[n]} I(Y_i,\CRS;Z) 
&=\sum_{i\in[n]} \big(I(\CRS;Z)+I(Y_i;Z|\CRS) \big) \\
&=\sum_{i\in[n]} I(Y_i;Z|\CRS)\\
& \geq I(Y_1,\dots,Y_n;Z|\CRS)\\ 
&\geq I(S;Z|\CRS) \\
& = I(S,\CRS;Z) - I(\CRS;Z) \\
& = I(S,\CRS;Z)\\
& \geq\frac{1}{2}\cdot \log|\calX|.
\end{align*}
Hence, there must exist an index $i^*$ such that $$I(Y_{i^*},\CRS;Z)\geq\frac{1}{n}\cdot I(S,\CRS;Z)\geq\frac{1}{2n}\cdot \log|\calX|.$$

We are now ready to complete the proof.
Observe that it suffices to prove the theorem assuming that $\epsilon=1$ and that $|\calX|=2^{\lambda(4en)^{\ell+1}}$. The reason is that any $(\epsilon,\delta)$-differentially private protocol with $\epsilon\leq1$ is also $(1,\delta)$-differentially private, and that a protocol for the common element problem over a domain $\calX$ is, in particular, a protocol for the common element problem over subsets of $\calX$.
By \cref{thm:MutualInfoBound} (our bound on the mutual information between the input and the messages submitted by any single party in a 1-round protocol), there exists a constant $\lambda>1$ such that
\begin{align*}
\frac{1}{2n}\cdot \log|\calX| &\leq I(Y_{i^*},\CRS;Z) \leq  \lambda\cdot (en)^{\ell}\cdot\left( \epsilon^2 + \frac{\delta}{\epsilon}\log|\calX| + \frac{\delta}{\epsilon}\log(\epsilon/\delta)  \right) + \ell\cdot \log\left(4en\right).
\end{align*}
Substituting $\epsilon=1$ and $|\calX|=2^{\lambda(4en)^{\ell+1}}$, 
and solving for $\delta$, we get that 
$
\delta\geq \frac{1}{8\lambda (en)^{\ell+1}}.
$
\end{proof}

\subsection{A two-round protocol with message complexity 1}
\label{sec:commonTwoRound}

Intuitively, \cref{thm:lowerBound} shows that in any 1-round protocol for the common element problem, we either have that the message complexity is large, or we have that $\delta$ cannot be too small.
In \cref{prot:CommonTwoRound} we present a two round protocol for the common element problem, in which the message complexity is 1 and $\delta$ can be negligible. Our protocol, which we call \ProtCommonTwoRound, uses the shuffle channel in only one of the two rounds, and the communication in the second round is done via a public channel. 

\begin{theorem}\label{thm:commontworounds}
Let $\delta\in(0,1)$.
\ProtCommonTwoRound, described in \cref{prot:CommonTwoRound}, is $(O(1),O(\delta))$-differentially private against coalitions of size $0.9n$ that solves the common element problem. The protocol uses two rounds (one via a public channel and one via the shuffle) and has message complexity 1.
\end{theorem}

\protocol{\ProtCommonTwoRound}{A two-round protocol in the shuffle model for the common element problem with message complexity 1.}{prot:CommonTwoRound}{
\textbf{Inputs:} Each party $P_i$ (for $i\in[n]$) holds an input $x_i\in\calX$. The analyzer $P_0$ has no input. All parties have access to a hash function $h:\calX\rightarrow[n^2/\delta]$ chosen with uniform distribution from a pairwise independent family (defined, e.g., using a public random string).
\begin{enumerate}
\item
Every party $P_i$ computes $y_i\leftarrow h(x_i)$.
\item\label{step:CommonTwoRoundsLDP}
The parties use the public channel to execute a 1-round $(\epsilon,0)$-differentially private protocol in the local model for histograms over the (distributed) database $Y=(y_1,y_2,\dots,y_n)$ with failure probability $\delta$ (see e.g.,~\cite{BunNS19}, or \cref{thm:LDPhistograms}). This results in a data structure $D$ (known to all parties) that gives estimations for the multiplicities of elements in $Y$. That is, for every $y\in[n^2/\delta]$ we have that $D(y)\approx\left|\left\{ i\in[n] : y_i=y \right\}\right|$.
\item\label{step:CommonTwoRoundsTerminate}
Let $y^*\in[n^2/\delta]$ be an element that maximizes $D(y)$. If $D(y)<\frac{98\cdot n}{100}$ then all parties terminate, and the analyzer outputs $\bot$.
\item
Otherwise, each party $P_i$ prepares a single message $m_i$ as follows:
\begin{enumerate}
    \item If $y_i\neq y^*$ then $m_i=\bot$.
    \item Otherwise, $m_i=\bot$ with probability $1/2$ and $m_i=x_i$ with probability $1/2$.
\end{enumerate}
\item
Each party $P_i$ sends the message $m_i$ to the shuffle. All parties receive a permutation $s$ of $(m_1,\dots,m_n)$.
\item\label{step:CommonTwoRoundsOutput}
The analyzer outputs the element $x^*\neq\bot$ with the largest number of appearances in $s$ (the analyzer fails if all elements of $s$ are equal to $\bot$).
\end{enumerate}
}

 We begin with the privacy analysis of \ProtCommonTwoRound.

\begin{lemma}\label{lem:CommonTwoPrivacy}
\ProtCommonTwoRound
 is $(O(1),O(\delta))$-differentially private against coalitions of size $0.9n$.
\end{lemma}

\begin{proof}
Fix an index $i\in[n]$, fix two $i$-neighboring input vectors $\vec{x}$ and $\vec{x'}$, and fix a coalition $\calC$ of size $|\calC|=0.9n$ such that $P_i\notin\calC$. We need to show that
$\View^\Pi_\calC(\vec{x}) \approx_{\epsilon,\delta} \View^\Pi_\calC(\vec{x}').$
First observe that with probability at least $1-\delta$ over the choice of the hash function $h$, we have that $h$ perfectly hashes all the different inputs in $\vec{x},\vec{x'}$ (note $\vec{x},\vec{x'}$ span at most $n+1$ different values). We proceed with the analysis after fixing such a hash function $h$.

We write $\vec{x}_{\calC}=\vec{x'}_{\calC}$ to denote the inputs of the parties in $\calC$, and fix the internal randomness $r_{\calC}$ of the parties in $\calC$. Now let $S_1$ and $S_2$ be random variables representing the output of the public channel and the shuffle, respectively, during the execution on $\vec{x}$, where we denote $S_2=\bot$ if the execution halted on \stepref{step:CommonTwoRoundsTerminate}. Similarly, $S'_1,S'_2$ denote the outputs of these channels during the execution on $\vec{x'}$. With these notations we have that

$$\View^\Pi_\calC(\vec{x})=\left(h,r_{\calC},\vec{x}_{\calC},S_1,S_2\right)
\qquad\text{and}\qquad
\View^\Pi_\calC(\vec{x}')=\left(h,r_{\calC},\vec{x}_{\calC},S'_1,S'_2\right).$$

Observe that $S_1$ and $S'_1$ are computed using an $(\eps,0)$-differentially private protocol in the local model (see \cref{thm:LDPhistograms}), and hence,
$$\left(h,r_{\calC},\vec{x}_{\calC},S_1\right)
\approx_{(\epsilon,0)}\left(h,r_{\calC},\vec{x}_{\calC},S'_1\right).$$

We next argue about $S_2$ and $S'_2$. For an element $x\in\calX$ we write $f_{\vec{x}}(x)$ to denote the number of occurrences of $x$ in the input vector $\vec{x}$. Also, let $x^*\in\calX$ denote the most frequent element in $\vec{x}$, that is, an element such that  $f_{\vec{x}}(x^*)$ is maximized.

\paragraph{Case (a)$\quad{\boldsymbol{ f_{\vec{x}}(x^*)\leq\frac{96\cdot n}{100}}}$ : }

By the utility guarantees of the protocol for histograms (executed on \stepref{step:CommonTwoRoundsLDP}), each of the two executions terminates in \stepref{step:CommonTwoRoundsTerminate} with probability at least $(1-\delta)$. This is because if $n=\Omega(\frac{1}{\eps^2}\log(\frac{1}{\eps\delta}))$ then with probability at least $(1-\delta)$ all of the estimates given by $D(\cdot)$ are accurate to within $\pm0.01n$ (see \cref{thm:LDPhistograms}). Therefore, in case (a) we have
\begin{align*}
\View^\Pi_\calC(\vec{x})&=\left(h,r_{\calC},\vec{x}_{\calC},S_1,S_2\right)
\approx_{(0,\delta)}
\left(h,r_{\calC},\vec{x}_{\calC},S_1,\bot\right)\\
&\approx_{(\eps,\delta)}
\left(h,r_{\calC},\vec{x}_{\calC},S'_1,\bot\right)
\approx_{(0,\delta)}
\left(h,r_{\calC},\vec{x}_{\calC},S'_1,S'_2\right)
=\View^\Pi_\calC(\vec{x}').
\end{align*}

\paragraph{Case (b)$\quad{\boldsymbol{ f_{\vec{x}}(x^*)>\frac{96\cdot n}{100}}}$ : }

Fix any value $s_1$ for the outcome of the public channel, such that all the estimates given by the resulting data structure $D(\cdot)$ are accurate to within $\pm0.01n$ w.r.t.\ $\vec{x}$. We first show that conditioned on such an $s_1$ we have that 
$$
\left(h,r_{\calC},\vec{x}_{\calC},s_1,S_2\right)
\approx_{(\eps,\delta)}
\left(h,r_{\calC},\vec{x}_{\calC},s_1,S'_2\right).
$$
To see this, observe that once we condition on $s_1$   
then either both executions terminate on \stepref{step:CommonTwoRoundsTerminate}, or in the two executions we have that $y^*=h(x^*)$ (because $f_{\vec{x}}(x^*)>0.96n$). If $s_1$ is such that the two executions terminate on \stepref{step:CommonTwoRoundsTerminate}, then (conditioned on $s_1$) we have $S_2=S'_2=\bot$ and so
$$
\left(h,r_{\calC},\vec{x}_{\calC},s_1,S_2\right)
\equiv
\left(h,r_{\calC},\vec{x}_{\calC},s_1,S'_2\right).
$$
Now suppose that the two executions do not halt prematurely, and that $y^*=h(x^*)$.
In that case, the outcome of the shuffle contains (randomly permuted) copies of $\bot$ and copies of $x^*$.
Note that since the outcome of the shuffle is randomly permuted, then the outcome distribution of the shuffle is determined by the number of occurrences of $x^*$.

Note that if $x_i$ and $x'_i$ are both equal to $x^*$, or are both different from $x^*$, then $S_2$ and $S'_2$ are identically distributed, which would complete the proof.
We, therefore, assume that exactly one of $x_i,x'_i$ is equal to $x^*$. Suppose without loss of generality that $x_i=x^*$ and $x'_i\neq x^*$.

Since $f_{\vec{x}}(x^*)>0.96n$ and since $|\calC|=0.9n$, there is a set of parties $\calI$ of size $|\calI|=0.05n$ such that
\begin{enumerate}
    \item $\calI\cap(\calC\cup\{i\})=\emptyset$.
    \item For every $j\in\calI$ we have that $x_j=x'_j=x^*$.
\end{enumerate}
We show that the outcome of the shuffle preserves differential privacy (over the randomness of the parties in $\calI$ and the randomness of the shuffle). Fix the randomness of all parties except for parties in $\calI$. Note that this fixes the messages that these parties submit to the shuffle, and suppose that party $P_i$ submits $x^*$ during the first execution and submits $\bot$ during the second execution (if party $P_i$ submits $\bot$ during both execution then the outcome of the shuffle is, again, identically distributed). Let $k$ denote the number of parties among the parties not in $\calI$ that submitted $x^*$ to the shuffle during the execution on $\vec{x}$. (So during the execution on $\vec{x'}$ exactly $k-1$ such parties submitted $x^*$.)

Let us denote by $Z$ the number of parties from $\calI$ that submits $x^*$ to the shuffle. Note that $Z\equiv\Bin\left(|\calI|,\frac{1}{2}\right)$.
By the Hoeffding bound, 
assuming that $n=\Omega(\ln(1/\delta))$ (large enough), with probability at least $1-\delta$ we have that
$\frac{9}{20}\cdot|\calI|\leq Z\leq\frac{11}{20}\cdot|\calI|$. In addition, by the properties of the Binomial distribution, for every $\frac{9}{20}\cdot|\calI|\leq z\leq\frac{11}{20}\cdot|\calI|$ we have that
$$
\frac{\Pr[Z=z]}{\Pr[Z=z+1]}
=\frac{2^{-|\calI|}\cdot\binom{|\calI|}{z}}
{2^{-|\calI|}\cdot\binom{|\calI|}{z+1}}
=\frac{z+1}{|\calI|-z}\in e^{\pm1}.
$$
Let us denote the number of occurrences of $x^*$ at the output of the shuffle during the two executions as $|S_2|$ and $|S'_2|$, respectively. So $|S_2|\equiv k+Z$ and $|S'_2|\equiv k-1+Z$.
Fix a set $F\subseteq[n]$ of possible values for $|S_2|$, and denote
$$
T=\{(f-k):f\in F\}
\qquad\text{and}\qquad
T'=\{(f-k+1):f\in F\}
$$
We have that
\begin{align*}
\Pr\left[|S_2|\in F\right]&=\Pr[Z\in T]
\leq\delta+\Pr\left[Z\in T\cap\left\{z: \frac{9|\calI|}{20}\leq z\leq\frac{11|\calI|}{20}\right\}\right]\\
&\leq\delta+e^1\cdot \Pr\left[Z-1\in T\cap\left\{z: \frac{9|\calI|}{20}\leq z\leq\frac{11|\calI|}{20}\right\}\right]\\
&\leq\delta+e^1\cdot \Pr\left[Z-1\in T\right]
=\delta+e^1\cdot \Pr\left[Z\in T'\right]
=\delta+e^1\cdot\Pr\left[\left|S'_2\right|\in F\right].
\end{align*}
A similar analysis shows that $\Pr\left[|S'_2|\in F\right]\leq\delta+e^1\cdot\Pr\left[|S_2|\in F\right]$. 
This shows that conditioned on an output of the public channel $s_1$ such that $D(\cdot)$ is accurate for $\vec{x}$, we have that
$$
\left(h,r_{\calC},\vec{x}_{\calC},s_1,S_2\right)
\approx_{(1,\delta)}
\left(h,r_{\calC},\vec{x}_{\calC},s_1,S'_2\right).
$$

So far, we have established that the outcome of the first round (that uses the public channel) preserves $(\eps,0)$-differential privacy, and, conditioned on the outcome of the first round being ``good'' (i.e., the resulting data structure $D$ is accurate) we have that the outcome of the second round (that uses the shuffle) preserves $(1,\delta)$-differential privacy. Intuitively, we now want to use composition theorems for differential privacy to show that the two rounds together satisfy differential privacy. A small technical issue that we need to handle, though, is that the privacy guarantees of the second round depend on the success of the first round. As the outcome of the first round is ``good'' with overwhelming probability, this technical issue can easily be resolved, as follows.

Consider two random variables $\tilde{S}_1$ and $\tilde{S'}_1$ that are identical to $S_1$ and $S'_1$, except that if the resulting data structure $D(\cdot)$ is {\em not} accurate, then the value is replaced such that the resulting data structure $D(\cdot)$ is exactly correct. Since the protocol for histograms fails with probability at most $\delta$, we have that
$$
\left(h,r_{\calC},\vec{x}_{\calC},\tilde{S}_1\right)
\approx_{(0,\delta)}
\left(h,r_{\calC},\vec{x}_{\calC},S_1\right)
\approx_{(\eps,\delta)}
\left(h,r_{\calC},\vec{x}_{\calC},S'_1\right)
\approx_{(0,\delta)}
\left(h,r_{\calC},\vec{x}_{\calC},\tilde{S'}_1\right).
$$
In words, consider an imaginary protocol in which the outcome distribution of the first round during the two executions is replaced by $\tilde{S}_1$ and $\tilde{S'}_1$, respectively.
The statistical distance between the outcome distribution of this imaginary protocol and the original protocol is at most $\delta$. In addition, for every possible fixture of the outcome of the first (imaginary) round we have the second round preserves differential privacy. Therefore, composition theorems for differential privacy show that the two rounds together satisfy differential privacy. Formally,
\begin{align*}
\View^\Pi_\calC(\vec{x})&=\left(h,r_{\calC},\vec{x}_{\calC},S_1,S_2\right)
\approx_{(0,\delta)}
\left(h,r_{\calC},\vec{x}_{\calC},\tilde{S}_1,S_2\right)\\
&\approx_{(1+\eps,\delta)}
\left(h,r_{\calC},\vec{x}_{\calC},\tilde{S'}_1,S'_2\right)
\approx_{(0,\delta)}
\left(h,r_{\calC},\vec{x}_{\calC},S'_1,S'_2\right)
=\View^\Pi_\calC(\vec{x}').
\end{align*}

\end{proof}

\begin{lemma}\label{lem:CommonTwoUtility}
\ProtCommonTwoRound solves the common element problem. 
\end{lemma}

\begin{proof}
Fix an input vector $\vec{x}=(x_1,\dots,x_n)\in\calX^n$ such that for every $i$ we have $x_i=x$.
By the utility guarantees of the locally-private protocol for histograms, with probability at least $1-\delta$ it holds that all of the estimates given by $D(\cdot)$ are accurate to within $\pm0.01n$. In that case, we have that $y^*$ (defined in \stepref{step:CommonTwoRoundsTerminate}) satisfies $y^*=h(x)$. Thus, every message submitted to the shuffle in the second round is equal to $x$ with probability 1/2, and otherwise equal to $\bot$. Therefore, the analyzer fails to output $x$ in \stepref{step:CommonTwoRoundsOutput} only if all of the parties submitted $\bot$ to the shuffle. This happens with probability at most $2^{-n}$. Overall, with probability at least $(1-\delta-2^{-n})$ the analyzer outputs $x$.
\end{proof}

\cref{thm:commontworounds} now follows by combining \cref{lem:CommonTwoPrivacy} and \cref{lem:CommonTwoUtility}.

\section{Possibility and Impossibility for the Nested Common Element Problem}
\label{sec:possibilityImpossibility}

In this section we define a nested version of the common element problem of \cref{def:CommonElementProblem}.
This problem has a parameter $0< \alpha <1$. We show that this problem cannot be solved in the shuffle model in one round with differential privacy against coalitions of size $\alpha n$
(regardless of the number of messages each party can send).
In contrast, we show that it can be solved with differential privacy in one round against coalitions of size $c n$ for any constant $c < \min\set{\alpha,1-\alpha}$ and in two rounds against coalitions of size $c n$ for any constant $c < 1$. The impossibility result for one round and the two round protocol imply  a strong separation between what can be solved  in one round  and in two rounds.

\begin{definition}[The nested common element problem with parameter $\alpha$]
Let $0 < \alpha < 1$.
Consider $n$ parties $P_1,\dots,P_n$ and an analyzer $P_0$ (as in \cref{rem:model}). The input of each party in $P_1,\dots,P_{\floor{\alpha n}}$ is an element $x_i \in \calX$ and the input of each party  $P_{\floor{\alpha n}+1},\dots,P_{n}$ is a vector $\vec{y_i}$ of $|\calX|$ elements from some finite domain $\calY$. The analyzer $P_0$ has no input. 
If all inputs of $P_1,\dots,P_{\floor{\alpha n}}$ are equal (i.e., $x_1=x_2=\cdots=x_{\floor{\alpha n}}$) and 
the $x_1$-th coordinate in all inputs of $P_{\floor{\alpha n}+1},\dots,P_{n}$ are equal (i.e., $\vec{y_{\floor{\alpha n}+1}}[x_1]=\vec{y_{\floor{\alpha n}+2}}[x_1]=\cdots=\vec{y_{n}}[x_1]$), then the analyzer $P_0$ must output $\vec{y_{\floor{\alpha n}+1}}[x_1]$ with probability at least $3/4$. The output is not restricted otherwise.
\end{definition}

\begin{remark}
When $|\calX|=\poly(n)$ and $|\calY|$ is at most exponential in $n$, then the length of the inputs of all parties is polynomial in $n$. Our impossibility result for the nested common element problem holds in this regime (specifically, when $|\calX| =\tilde{\Omega}(n^2)$ and  $|\calY|=2$). Our protocols are correct and private regardless of the size of $\calX$ and $\calY$.
\end{remark}

In this section, we prove the following three theorems. 

\begin{theorem}
\label{thm:NestedImpossibility}
Let $|\calX| = \tilde{\Omega}(n^2) $. There is no one-round $(1,o(1/n))$-differentially private protocol in the shuffle model against coalition of size 
$\floor{\alpha n}$ for the nested common element problem with parameter $\alpha$ (regardless of the number of messages each party can send).
\end{theorem}

\begin{theorem}
\label{thm:NestedTwoRound}
For every  $0 < c < 1$, $\epsilon,\delta\in [0,1]$, and $n \geq \frac{200}{(1-c)n}\ln \frac{4}{\delta}$ there exists a two-round $(\epsilon,\delta)$-differentially private protocol against coalitions of size $cn$ that with probability at least $1-1/2^{n-1}$ solves the nested common element problem with parameter $\alpha$. 
\end{theorem}

\begin{theorem}
\label{thm:NestedOneRound}
For every constants $c,\alpha$ such that $0 < c < \min\set{\alpha,1-\alpha} <1$, there exists a constant $\epsilon_0$ such that  there exits a one-round $(\epsilon_0,\delta)$-differentially private protocol against coalitions of size $cn$ that with probability at least $3/4$ solves the nested common element problem with parameter $\alpha$, where $\delta=2^{-O(\min\set{\alpha,1-\alpha}-c)n)}$ and
$n\geq 6 \cdot \max\set{1/\alpha,1/(1-\alpha)}$. 
\end{theorem}

\subsection{An impossibility result for private one-round protocols for the nested common element problem}

\label{sec:nestedImpossibility}

We next show that the nested common element problem with parameter $\alpha$ cannot be solved privately against coalitions of size $\alpha n$ when $\calX$ is large enough, namely, when $|\calX|=\tilde{\Omega}(n^2)$.
The proof of the impossibility result is done by using an impossibility result to the vector common element problem (in the centralized model) defined below.
\begin{definition}[The vector common element problem] 
The input of  the problem is a database containing  $n$ vectors $(\vec{y_1},\ldots,\vec{y_n})\in (\set{0,1}^d)^n$. 
For a given set of vectors $\vec{y_1},\dots,\vec{y_n}$, define for every $b\in \set{0,1}$
$$I_b=\set{j: \vec{y_1}[j]=\cdots = \vec{y_n}[j]=b}.$$ 
To solve the the vector common element problem,  
an analyzer  must output with probability 
at least $1-o(1/n)$ sets $J_0$ and $J_1$ such that $I_0\subseteq J_0$,  $I_1\subseteq J_1$, and $J_0\cap J_1=\emptyset$. 
\end{definition}
In words, the task in the vector common element problem is to identify the coordinates in which the inputs vectors agree, that is, for each coordinate if all the vectors agree on the value of the coordinate then the algorithm should return this coordinate and the common value; if the vectors do not agree on this coordinate then the algorithm can say that this is either a zero-coordinate, a one-coordinate, or none of the above.

The following theorem is implied by the techniques of~\cite{BunUV18} (i.e., the reduction to fingerprinting codes).
\begin{theorem}[\cite{BunUV18}]
\label{thm:VectorCommonElement}
For every $d \in \NN$, any $(1,o(1/n))$-differentially private algorithm in the centralized model  for the vector common element problem with vectors of length $d$ has sample complexity 
$\tilde{\Omega}(\sqrt{d})$.
\end{theorem}

We next prove our impossibility result, i.e., prove \cref{thm:NestedImpossibility}.

\begin{proof}[Proof of \cref{thm:NestedImpossibility}]
We show that if for $|\calX| = \tilde{\Omega}(n^2)$ there is an $n$-party protocol, denoted $\Pi$, in the shuffle model for the nested common element problem  with parameter $\alpha$ that is private against the coalition of parties holding the $x$-inputs, namely, $\calC=\set{P_1,\ldots,P_{\floor{\alpha n}}}$, then there is an algorithm in the centralized model for the vector common element problem with database of size $O(n^2 \log n)$ violating \cref{thm:VectorCommonElement}.

As a first step, consider the following algorithm $\calA_1$  for the vector common element problem  in the centralized model,
whose inputs are $\vec{y_{\floor{\alpha n}+1}},\ldots,\vec{y_n}$ (each vector of length $|\calX|$).
\begin{enumerate}
\item 
The analyzer chooses a public random string $\crs$.
\item
For each  $i \in \set{{\floor{\alpha n}+1},\ldots,n}$, the analyzer simulates party $P_i$  in protocol $\calP$  with the input  $\vec{y_i}$ and the public random string $\crs$, generating a vector of messages $\vec{m_i}$.
\item
The analyzer shuffles the messages in  $\vec{m_{\floor{\alpha n}+1}},\cdots,\vec{m_n}$, denote the output of the shuffle by $\vec{\tilde{m}}$.
\item
For every $x \in \calX$ do:
    \begin{enumerate}
    \item 
    For each  $i \in \set{{1},\ldots,{\floor{\alpha n}}}$, the analyzer simulates party $P_i$ in protocol $\calP$ with the input $x$ and the public random string $\crs$, generating a vector of messages $\vec{m_i}$.
    \item
    \label{step:calPanalyzer}
    The analyzer shuffles the messages in  $\vec{\tilde{m}},\vec{m_1},\ldots,\vec{m_{\floor{\alpha n}}}$,
    gives the shuffled messages to the analyzer of $\calP$, and gets an output $z_x$.
    \end{enumerate}
\item
The analyzer returns $I_b=\set{x:z_{x}=b}$ for $b \in \set{0,1}$.
\end{enumerate}

First we argue that $\calA_1$ is $(1,o(1/n))$-differentially private: The coalition $\calC$ sees the output of the shuffle in $\calP$ and can remove the messages it sent to the shuffle  in $\calP$,
therefore computing $\vec{\tilde{m}}$ from the view is a post-processing of an $(\epsilon,o(1/n))$-differentially private output.
Second, notice that for every $x\in \calX$,
the shuffled messages that the analyzer of $\calP$ gets in \stepref{step:calPanalyzer}
are distributed as in $\calP$, thus, if $\vec{y_{\floor{\alpha n}+1}}[x]=\cdots=\vec{y_{n}}[x]=b$,
then $z_x=b$ with probability at least $3/4$ (however for $x\neq x'$ these events might be independent). 

The success probability of $\calA_1$ is not enough to violate \cref{thm:NestedImpossibility} and we repeat it $O(\log |\calX|)$ times. This is done in $\calA_2$, which preserves the privacy using sub-sampling:
\begin{enumerate}
\item \textbf{Inputs:} vectors $\vec{y_1},\ldots,\vec{y_{t}}$, where  $t=O(n\ln|\calX|)$.
\item For $\ell=1$ to $4\ln |\calX|$ do:
    \begin{enumerate}
    \item 
    \label{step:sample}
    Sample a set $T \subset [t]$ of size $\frac{t}{(3+\exp(1))4\ln |\calX|}=n$ and execute $\calA_1$ on the vectors $(\vec{y_i})_{i \in T}$ and get sets $J_0^\ell,J_1^\ell$.
    \end{enumerate}
\item For $b \in \set{0,1}$, let $J_b = \set{j : j\in J_b^\ell \text{ \rm for more than $4 \ln |\calX|$ indices } \ell }.$ 
\end{enumerate}
By \cref{clm:boostPrivacy} (i.e., sub-sampling) and since $\calA_1$ is $(1,o(\frac{1}{n}))$-differentially private, each execution of \stepref{step:sample} is $(\frac{1}{4\ln |\calX|},o(\frac{1}{n \ln |\calX|}))$-differentially private. By simple composition, algorithm $\calA_2$ is $(1,o(1/n))$-differentially private.

We next argue that with probability at least $1-o(1/n)$ algorithm $\calA_2$ outputs disjoint sets $J_0,J_1$ such that $I_0 \subseteq J_0$ and $I_1\subseteq J_1$. Fix $j$ such that $\vec{y_1}[j]=\cdots=\vec{y_t}[j]=b$ for some $b$. By the correctness of $\calA_1$, for every $\ell \in [4\ln |\calX|]$ it holds that $j \in J_b^\ell$ with probability at least $3/4$ and these events are independent. Thus, by the Hoeffding inequality, $j \in J_b^\ell$ for more than half of the values of $\ell$ with probability at least $1-1/|\calX|^2$. By the union bound, the probability that the algorithm errs for some coordinate for which all vectors $\vec{y_i}$ agree is at most $1/|\calX|=\tilde{O}(1/n^2)=o(1/n)$. 

To conclude, assuming that $\calP$ as above exits, we constructed a $(1,o(1/n))$-differentially private algorithm $\calA_2$ with database of size $O(n^2 \log n)$ and $d=|\calX|=\tilde{\Omega}(|\calX|^2)$,
contradicting \cref{thm:VectorCommonElement}.
\end{proof}

\subsection{Private protocols for the nested common element problem}
\label{sec:nestedPossibility}

\subsubsection{Prelude}
As a warm-up, we present \ProtCommonPrelude (described in \cref{prot:CommonPrelude}), a new  protocol for the common element problem of \cref{def:CommonElementProblem}, i.e., there are $n$ parties, each holding an element $x_i$. If $x_1=\cdots=x_n$, then the analyzer must output $x_1$. (There are more efficient protocols for this problem.) We then present \ProtNestedCommon (described in  \cref{prot:NestedCommon}), which generalizes \ProtCommonPrelude to the nested functionality and we use the proof of privacy of the former protocol to prove the privacy of the latter protocol.

\protocol{\ProtCommonPrelude}{A one-round protocol in the shuffle model against coalitions of size at most $cn$ for the common element problem (for some constant $c<1$).}{prot:CommonPrelude}{
\begin{enumerate}
\item
Let $G$ be any  additive group such that $|G| \geq 16 |\calX|$.  
\item
Each party $P_i$  prepares a vector $\vec{z_i}$ of length $|\calX|$ of elements from $G$ as follows:
\begin{enumerate}
\item 
\label{step:participate}
With probability $3/4$ (party $P_i$ participates):
\begin{enumerate}
    \item For every $x\neq x_i$ let $\vec{z_i}[x]$ be a random element in $G$ independently chosen with uniform distribution.
    \item 
    \label{step:noise}
    With probability $1/6n$ (sets location $x_i$ to noise):
    Let $\vec{z_i}[x_i]$ be a random element in $G$ independently chosen with uniform distribution.
    \item 
    \label{step:tru}
    With probability $1-\frac{1}{6n}$ (sets location $x_i$ to zero): Let $\vec{z_i}[x_i]=0$.
\end{enumerate}
\item
\label{step:dontparticipate}
With probability $1/4$ (party $P_i$ does not participate): 
For every $x \in \calX$ let $\vec{z_i}[x]=0$. 
\end{enumerate}
\item
The parties execute the $\delta'$-secure protocol of \cref{thm:IKOS} for addition over $G^{|\calX|}$ in the shuffle model where the input of $P_i$ is $\vec{z_i}$. Let $\vec{z}$ be the sum. 
\item
If there is a unique coordinate $x$ such that $\vec{z}[x]=0$, then the analyzer outputs $x$.
\end{enumerate}
}

\begin{lemma}
\label{lem:ProtCommonPrelude}
Let $c <1$ be a constant and $G$ be a group with at least $16|\calX|$ elements.
There exists a constant $\epsilon_0 =O(\ln \frac{1}{1-c})$ such that for every $n>2$ \ProtCommonPrelude is an $(\epsilon_0,\delta)$-differentially private protocol against coalitions of size $cn$ that solves the common element problem with probability at least $3/4$, where $\delta=O(e^{-(1-c)n/8})$. 
\end{lemma}
\begin{proof}
We first prove the correctness of the protocol; we only need to consider the case when $x_1=\cdots=x_n$. We say that if any party executes \stepref{step:noise}, then the protocol fails. By the union bound, the probability of this event is at most $1/8$. If this event does not occur, then $\vec{z}[x_1]=0$ and the protocol fails if there is a $x\neq x_1$ such that $\vec{z}[x]=0$. This can happen in two cases: (1) no party participates -- this occurs with probability $4^{-n}$, which is less than $1/16$ for $n \geq 2$. (2) the sum of all random elements in some coordinate is 0.
This occurs with probability at most $|\calX|/|G|$. As $G$ is a group with at least $16|\calX|$ elements, the last probability is less than $1/16$. All together the probability of failure is less than $1/4$.

We next provide the privacy analysis. 
We assume an ideal functionality for addition, that is, the parties in a coalition $\calC$ only see the sum $\vec{z}$ and $(\vec{z_i})_{P_i\in \calC}$. We use a statistically-secure protocol for addition in the shuffle model. By \cref{lem:composition}, this statistical security will add to the $\delta$ in the $(\epsilon,\delta)$-differential privacy.

Let $\calC$ be a coalition of size $cn$. W.l.o.g., assume that $\calC=\set{P_{(1-c)n+1},\ldots,P_n}$  and  consider two databases that differ on $x_1$.
Let $n'=(1-c)n$ and  $\vec{z'}=\sum_{i=1}^{n'} \vec{z_i}$. Note that the coalition $\calC$ can compute $\vec{z'}$, but has no information on  $(\vec{z_i})_{1 \leq i \leq n'}$.
We say that a party $P_i \in \notcalC$ does not participate in the protocol if it executes \stepref{step:dontparticipate}, otherwise we say that it participates. We start with two observations:
\begin{itemize}
    \item The probability that more than $1/2$ of the parties in $\notcalC$ do not participate is less than $e^{-n'/8}$ (by the Hoeffding bound). In this case privacy may fail, and this will fall under the $\delta$ in the definition of $(\epsilon,\delta)$-differential privacy.
    \item If the  parties in  $ \set{P_2,\dots,P_{n'}}$ that participate do not hold the same value $x$, then the output $\vec{z'}$ is a uniformly random vector, regardless of $P_1$’s action (and input), and privacy holds. 
\end{itemize}

In the rest of the proof we assume that the above two events do not hold and fix the set of parties $T \subseteq \set{P_2,\dots,P_{n'}}$ that participate (this set is of size at least $n'/2$). The only two possible cases for the output: (1) a random vector $\vec{z'}$ and (2) a vector $\vec{z'}$ in which $\vec{z'}[x]=0$ and all coordinates except for  $\vec{z'}[x]$ are random. We show that the probabilities of both possible cases for the output are bigger than some constant $0< p < 1$ when $x_1=x$ and when $x_1\neq x$.
The probability that at least one party
in $T$ executes \stepref{step:noise} is at least 
$$1-(1-1/6n)^{|T|} \geq 1-(1-1/6n)^{n'/2} \approx 1- e^{-(1-c)/12}\approx \frac{1-c}{12}$$  (regardless of the input of $P_1$); in this case    $\vec{z'}$ is a uniformly random vector.
Otherwise, if $P_1$ does not participate, the output is $\vec{z'}[x]=0$ (even if $x_1 \neq x$). 
The probability that no party in $T$ executes \stepref{step:noise} is 
$(1-1/6n)^{|T|} \geq (1-1/6n)^{n'}\approx e^{-(1-c)/6} \approx 1-  \frac{1-c}{6} \geq \frac{5}{6}$. Thus, the probability that no party in $T$ executes \stepref{step:noise} and $P_1$ does not participate is at least $5/6 \cdot 1/4=5/24$.
\end{proof}

\subsubsection{A private one-round protocol for the nested common element problem}
\label{sec:nestedPossibility1round}

We next present for every constant $c < \min\set{\alpha,1-\alpha}$, a  protocol, called \ProtNestedCommon, for the nested common element problem with parameter $\alpha$ that is private for coalitions of size $cn$. A possible idea to construct such protocol is to execute \ProtCommonPrelude for every coordinate $\vec{y_i}[x]$ of the vectors. The problem with this idea is that the analyzer will learn the values of $\vec{y_i}$ for all coordinates the vectors agree. By the proof of \cref{thm:NestedImpossibility} this is impossible (when $\calX$ is big). The solution is that each party $P_i$ that  holds an input $x_i$ adds noise for every $x\neq x_i$ (that is, with probability 1 it sends a random vector in the execution of \ProtCommonPrelude for coordinate $x$). Similar to \ProtCommonPrelude, with some small probability it has to add noise also to the  execution of \ProtCommonPrelude for coordinate $x_i$. In \cref{lem:NestedOneRound}, we  prove the correctness and privacy of \ProtNestedCommon, proving \cref{thm:NestedOneRound}.

\protocol{\ProtNestedCommon}{A one-round protocol in the shuffle model for the nested common element problem.}{prot:NestedCommon}{
\begin{enumerate}
\item
Let $G$ be any additive group such that  $|G| \geq 16|\calX||\calY|$. 
\item
Each party $P_i$  for $1 \leq i \leq \floor{\alpha n}$ prepares a vector $\vec{z_i}$ of length $|\calX| \cdot |\calY|$ of elements from $G$ as follows:
\begin{enumerate}
\item 
\label{step:x-participate}
With probability $3/4$ (party $P_i$ participates):
\begin{enumerate}
    \item For every $x\neq x_i$ and $y \in \calY$ let $\vec{z_i}[x,y]$ be a random element in $G$ independently chosen with uniform distribution.
    \item 
    \label{step:x-noise}
    With probability $1/6n$ (sets $x_i$-locations to noise):
    For every $y \in \calY$ let $\vec{z_i}[x_i,y]$ be a random element in $G$ independently chosen with uniform distribution.
    \item 
    \label{step:x-tru}
    With probability $1-1/6n$ (sets $x_i$-locations to zero): For every $y \in \calY$ let $\vec{z_i}[x_i,y]=0$.
\end{enumerate}
\item
\label{step:x-dontparticipate}
With probability $1/4$ (party $P_i$ does not participate): 
For every $x \in \calX$ let $\vec{z_i}[x]=0$. 
\end{enumerate}
\item
Each party $P_i$  for $\floor{\alpha n}+1 \leq i \leq n$ prepares a vector $\vec{z_i}$ of length $|\calX| \cdot |\calY|$ of elements from $G$ as follows:
\begin{enumerate}
\item 
\label{step:y-participate}
With probability $3/4$ (party $P_i$ participates):
\begin{enumerate}
    \item For every $x \in \calX$ and every $y\neq \vec{y_i}[x]$ let $\vec{z_i}[x,y]$ be a random element in $G$ independently chosen with uniform distribution.
    \item 
    \label{step:y-noise}
    With probability $1/6n$ (sets $y_i$-locations to noise):
    For every $x \in \calX$, let $\vec{z_i}[x,\vec{y_i}[x]]$ be a random element in $G$ independently chosen with uniform distribution.
    \item 
    \label{step:y-tru}
    With probability $1-1/6n$ (sets $y_i$-locations to zero): For every $x \in \calX$, let $\vec{z_i}[x,\vec{y_i}[x]]=0$.
\end{enumerate}
\item
\label{step:y-dontparticipate}
With probability $1/4$ (party $P_i$ does not participate): 
For every $x \in \calX$ let $\vec{z_i}[x]=0$. 
\end{enumerate}
\item
The parties execute the $\delta'$-secure protocol of \cref{thm:IKOS} for addition over $G^{|\calX|}$ in the shuffle model where the input of $P_i$ is $\vec{z_i}$. Let $\vec{z}$ be the sum. 
\item
If there is a unique coordinate $x$ such that $\vec{z}[x,y]=0$, then the analyzer outputs $y$.
\end{enumerate}
}

\begin{lemma}
\label{lem:NestedOneRound}
Let $G$ be a group with at least $16|\calX| \cdot |\calY|$ elements.
For every constants $c,\alpha$ such that $0 < c < \min\set{\alpha,1-\alpha} <1$, there exists a constant $\epsilon_0 > 1 $ such that for every $n\geq 6 \cdot \max\set{1/\alpha,1/(1-\alpha)}$ \ProtNestedCommon is a one-round $(\epsilon_0,\delta)$-differentially private protocol against coalitions of size $cn$ that with probability at least $3/4$ solves the nested common element problem with parameter $\alpha$, where $\delta=O\left(2^{(\min\set{\alpha,1-\alpha}-c)n}\right)$. 
\end{lemma}

\begin{proof}
We first prove the correctness of \ProtNestedCommon (similar to the correctness of \ProtCommonPrelude); we only need to consider the case when $x_1=\cdots=x_{\floor{\alpha n}}$ and $\vec{y_{\floor{\alpha n}+1}}[x_1]=\ldots=\vec{y_n}[x_1]$ (there are no correctness requirements if these condition do not hold). We say that if any party executes \stepref{step:x-noise} or \stepref{step:y-noise}, then the protocol fails. By the union bound, the probability of this event is at most $1/8$. If this event does not occur, then $\vec{z}[x_1,\vec{y_1}[x_1]]=0.$ In this case the protocol fails if there is a $x,y\neq x_1,\vec{y_1}[x_1]$ such that $\vec{z}[x,y]=0$. This can happen in two cases: (1) no $x$-party participates or no $\vec{y}$-party participates -- this occurs with probability $4^{-{\alpha n}}+4^{-{(1-\alpha) n}}$, which is less than $1/16$ for $n \geq 6 \cdot \max\set{1/\alpha,1/(1-\alpha)}$. (2) the sum of all random elements in some coordinate is 0.
This occurs with probability at most $|\calX||\calY|/|G|$. As $G$ is a group with at least $16|\calX|\cdot |\calY|$ elements, the last probability is less than $1/16$. All together the probability of failure is less than $1/4$.

We next prove the privacy. Again, we assume an ideal functionality for addition and consider two neighboring databases. 
Let $\calC$ be a coalition of size $cn$ and $\vec{z'}=\sum_{i:P_i \notin \calC}$. 
If there is no party in $\set{P_1,\dots,P_{\floor{\alpha n}}}\setminus \calC$ (i.e., an honest $x$-party) that participates or
there is no party in $\set{P_{\floor{\alpha n}+1},\dots,P_{n}}\setminus \calC$ (i.e., an honest $y$-party) that participates
(which occurs with probability  $4^{-(\alpha-c)n}+4^{-(1-\alpha-c)n}$), then we say that the privacy fails and we pay for it in the $\delta$.
There are two cases. 
\begin{itemize}
\item There exists an $1\leq i \leq \floor{\alpha n}$ such that $x_i\neq x'_i$. In this case  $(\vec{z_i})_{i:P_i \notin \calC}$ are simply  vectors that $(P_i)_{i:P_i \notin \calC}$ send to the shuffle in an execution of \ProtCommonPrelude over the group $G^{|\calY|}$, and the privacy follows from 
\cref{lem:ProtCommonPrelude}. Note that in the proof of \cref{lem:ProtCommonPrelude} we only need the following properties: (1) at least one of the parties from $\set{P_1,\dots,P_{\floor{\alpha n}}}\setminus \calC$ participates, and (2) each party not in $\calC$  sends a random vector with probability $1/8$ and this is true in \ProtNestedCommon regardless if the party is an $x$-party or a $\vec{y}$-party.
\item
There exists an $\floor{\alpha n}+1\leq i \leq n$ such that $\vec{y_i}\neq \vec{y'_i}$. 
We assumed that there is at least one honest $x$-party that participates.
If there are two honest $x$-parties that participate with a different input, then $\vec{z'}$ is a random vector regardless of the input of $P_i$. Otherwise, let $P_{i'}$ be an honest $x$-party that participates. Then, all coordinates $\vec{z}[x,y]$, where $x \neq x_i$ are random elements.  In this case we can ignore all entries in the vector $\vec{z'}[x',y]$, where $x' \neq x$ and effectively the view of $\calC$ is the view   in an execution of \ProtCommonPrelude over the group $G^{|\calY|}$, and the privacy follows from  \cref{lem:ProtCommonPrelude} (since we assumed that there is at least one honest $y$-party that participates). \qedhere
\end{itemize}
\end{proof}

\begin{remark}
In \ProtNestedCommon (as well as \ProtCommonPrelude) the privacy parameter is $\epsilon_0 > 1$.
Using sub-sampling (\cref{clm:boostPrivacy}) we can reduce the privacy parameter to any $\epsilon$ by
increasing the number of parties by a multiplicative factor of $O(1/\epsilon)$. Notice that this is possible in \ProtNestedCommon since it only uses the shuffle and all $x$-parties (respectively,  all $\vec{y}$-parties) are symmetric. 
\end{remark}

\subsubsection{A private two-round protocol for the nested common element problem}
\label{sec:nestedPossibility2round}

We next present a two-round differentially private protocol for the nested common element problem against a coalition of size $cn$ for every $c < 1$.
In our protocol we will need a one-round protocol for a variant common element problem, called the $\alpha$-common element problem, where only $P_1,\dots,P_{\floor{\alpha n}}$ hold inputs, that is, if $x_1=x_2=\cdots=x_{\floor{\alpha n}}$, then with probability at least $3/4$ the analyzer must output $x_1$. Of course,  $P_1,\dots,P_{\floor{\alpha n}}$ can execute \ProtCommonPrelude to solve this problem, however the protocol will be private only against coalitions of size $cn$, where $c< \alpha$. 
To achieve this goal, we use a protocol of Balcer and Cheu~\cite{BC20} for histograms, which reports all elements that appear frequently in a database, and in particular, the $\alpha$-common element. The properties of the protocol are summarized in~\cref{thm:BC20}. Note that  coalitions are not onsidered in~\cite{BC20}, however the view of a coalition of size $cn$ in this protocol is basically the view of the analyzer in a protocol with $(1-c)n$ parties, hence the security of their protocol against coalitions follows (taking a slightly bigger $\epsilon$).  
\begin{theorem}[Special case of~\protect{\cite[Theorem 12]{BC20}}]
\label{thm:BC20}
For every $c<1$, $\epsilon,\delta\in [0,1]$, and $n \geq \frac{200}{(1-c)n}\ln \frac{4}{\delta}$
there is a one-round $(\epsilon,\delta)$-differentially private protocol against coalitions of size $cn$ that solves the $\alpha$-common element problem  with probability at least $1-1/2^{n}$. The message complexity of this protocol is $O(|\calX|)$.
\end{theorem}

We next prove the existence of a two round protocol for the nested common element problem, proving \cref{thm:NestedTwoRound}.
\begin{proof}[Proof of \cref{thm:NestedTwoRound}]
The protocol is the natural protocol. We first execute the $\alpha$-common element problem to find the value $x$ that is common among the parties $P_1,\dots,P_{\floor{\alpha n}}$ (assuming such value exists). If the protocol returns some value $x_0$, then in the second round we execute the $(1-\alpha)$-common element problem to find the value $y$ that is common among the parties $P_{\floor{\alpha n}+1},\dots,P_n$ when holding the elements  $\vec{y_{\floor{\alpha n}+1}}[x_0],\dots,\vec{y_{n}}[x_0]$ (assuming such value exists). If this protocol returns a value $y_0$, then return this value.
\end{proof}

\section*{Acknowledgments}
The authors thank Rachel Cummings and Naty Peter for discussions of the shuffle model at an early stage of this research.
Work of A.~B.\ and K.~N.\ was supported by NSF grant No.~1565387 TWC: Large: Collaborative: Computing Over Distributed Sensitive Data. This work was done when A.~B.\ was hosted by Georgetown University.
Work of A.~B.\ was also supported by Israel Science Foundation grant no.~152/17, a grant from the Cyber Security Research Center at Ben-Gurion University, and ERC grant 742754 (project NTSC). I.~H.\ is the director of the  Check Point Institute for Information Security. His research is supported by ERC starting grant 638121 and Israel Science Foundation grant no.  666/19. Work of U.~S.\ was supported in part 
by the Israel Science Foundation (grant 1871/19), and by the Cyber Security Research Center at Ben-Gurion University of the Negev.

\bibliographystyle{plain}

\appendix

\section{Preliminaries from Information Theory}
\label{app:info}

We recall basic definitions from information theory. The following definition was introduced by Shannon as a measure for the uncertainty in a random variable.

\begin{definition}[Entropy]
The {\em entropy} of a random variable $X$ is
$$
H(X) = - \sum_x \Pr[X=x]\cdot\log\left(\Pr[X=x]\right).
$$
The {\em joint entropy} of (jointly distributed) random variables $X$ and $Y$ is
$$
H(X,Y) = - \sum_{x,y} \Pr[X=x,Y=y]\cdot\log\left(\Pr[X=x,Y=y]\right).
$$
Note that this is simply the entropy of the random variable $Z=(X,Y)$.
\end{definition}

\begin{remark}\label{rem:entropy}
Observe that, by definition, we have that
\begin{align*}
H(X)=-\E_X\Big[ \log(\Pr[X=x]) \Big]&&\text{and}&&
H(X,Y)=-\E_{(X,Y)}\Big[ \log(\Pr[X=x,Y=y]) \Big].
\end{align*}
\end{remark}

We will use the following bounds on the entropy of a binary random variable.

\begin{claim}\label{claim:binaryHbounds}
Let $X$ be a random variable such that $\Pr[X=1]=p$ and $\Pr[X=0]=1-p$ for some $p\in[0,1]$. Then
$$
p\cdot\log\frac{1}{p} \leq H(X)\leq p\cdot\log\frac{4}{p}.
$$
\end{claim}

\begin{proof}
The lower bound is immediate from the definition of $H(X)$. The upper bound follows from the fact that for every $x\in[0,1]$ it holds that $-(1-x)\log(1-x)\leq 2x$. Specifically,
\begin{align*}
H(X)&=-p\log(p)-(1-p)\log(1-p)\leq -p\log(p)+2p=p\cdot\log\frac{4}{p}.
\end{align*}
\end{proof}

The following definition can be used to measure the uncertainty in one random variable conditioned on another.

\begin{definition}[Conditional entropy]
Let $X,Y$ be two random variables. The {\em conditional entropy} of $X$ given $Y$ is
$$H(X|Y)=\E_{Y}\Big[H(X|Y=y)\Big],$$
where
$$
H(X|Y=y) = - \sum_x \Pr[X=x|Y=y]\cdot\log\left(\Pr[X=x|Y=y]\right).
$$
Similarly, given another random variable $Z$, we have that
$$H(X|Y,Z=z)=\E_{Y|Z=z}\Big[H(X|Y=y,Z=z)\Big].$$
\end{definition}

\begin{remark}\label{rem:conditionalEntropy}
Observe that
\begin{align*}
H(X|Y)&=\E_{Y}\Big[H(X|Y=y)\Big]\\
&=\sum_{y}\Pr[Y=y]\cdot H(X|Y=y)\\
&=-\sum_{y}\Pr[Y=y]\cdot\sum_{x}\Pr[X=x|Y=y]\cdot\log(\Pr[X=x|Y=y])\\
&=-\sum_{x,y}\Pr[X=x,Y=y]\cdot\log(\Pr[X=x|Y=y])\\
&=-\E_{(X,Y)}\log(\Pr[X=x|Y=y]).
\end{align*}
\end{remark}

Consider two random variables $X$ and $Y$. The following Claim shows that, intuitively, the uncertainty in $(X,Y)$ is the uncertainty in $X$ plus the uncertainty in $Y$ given $X$.

\begin{claim}[Chain rule of conditional entropy]\label{claim:chainCondEnt}
Let $X,Y$ be two random variables. Then,
$$H(X,Y) =  H(X)+H(Y|X).$$
\end{claim}

\begin{proof}
For every $x,y$ we have that
$$
\Pr[X=x,Y=y]=\Pr[X=x]\cdot\Pr[Y=y|X=x],
$$
and hence,
$$
\log(\Pr[X=x,Y=y])=\log(\Pr[X=x])+\log(\Pr[Y=y|X=x]).
$$
Taking the expectation over $X$ and $Y$ we get that
$$
\E_{(X,Y)}\Big[\log(\Pr[X=x,Y=y])\Big]=\E_{X}\Big[\log(\Pr[X=x])\Big]+\E_{(X,Y)}\Big[\log(\Pr[Y=y|X=x])\Big].
$$
Therefore, by \cref{rem:entropy,rem:conditionalEntropy} we have that
$$H(X,Y) =  H(X)+H(Y|X).$$
\end{proof}

Consider two random variables $X$ and $Y$.
The following definition can be used to measure the amount of ``information'' that $X$ gives on $Y$.

\begin{definition}[Mutual information]
Let $X,Y$ be two random variables. The {\em mutual information} of $X$ and $Y$ is
$$
I(X;Y) = H(X) - H(X|Y)
$$
\end{definition}

\begin{remark}
Observe that by the chain rule of conditional entropy (\cref{claim:chainCondEnt}), we get that the mutual information that $X$ gives about $Y$ equals the mutual information that $Y$ gives about $X$. Formally,
\begin{align*}
I(X;Y) &= H(X)-H(X|Y)\\
&= H(X) - H(X,Y) + H(Y)\\
&= -H(Y|X)+H(Y)\\
&=I(Y;X).
\end{align*}
\end{remark}
\begin{claim}
Let $X$ and $Y$ be random variables. Then,
$$0 \leq I(X;Y) \leq H(X).$$
\end{claim}
The following definition can be used to measure the amount of ``information'' that $X$ gives on $Y$, conditioned on a third random variable $Z$.

\begin{definition}[Conditional mutual information]
Let $X,Y,Z$ be three random variables. The {\em conditional mutual information} of $X$ and $Y$ given $Z=z$ is
$$
I(X;Y|Z=z) = H(X|Z=z) - H(X|Y,Z=z)
$$
The {\em conditional mutual information} of $X$ and $Y$ given $Z$ is
$$I(X;Y|Z)=\E_{Z}\Big[I(X;Y|Z=z)\Big].$$
\end{definition}

The following claim gives an alternative definition for conditional mutual information.

\begin{claim}
$$
I(X;Y|Z) = H(X|Z) - H(X|Y,Z).
$$
\end{claim}

\begin{proof}
\begin{align*}
I(X;Y|Z) &= \E_{Z}\Big[I(X;Y|Z=z)\Big]\\
&= \E_{Z}\Big[H(X|Z=z)\Big] - \E_{Z}\Big[H(X|Y,Z=z)\Big]\\ 
&= H(X|Z) - \E_{Z}\left[\E_{Y|Z=z}\Big[H(X|Y=y,Z=z)\Big]\right]\\
&= H(X|Z) - \E_{(Y,Z)}\Big[H(X|Y=y,Z=z)\Big]\\
&= H(X|Z) - H(X|Y,Z).
\end{align*}
\end{proof}

\begin{claim}[Chain rule for mutual information]\label{claim:chainMutualInfo}
$$I(X,Y;Z) = I(X;Z) + I(Y;Z|X).$$
\end{claim}

\begin{proof}
Using the chain rule for conditional entropy (\cref{claim:chainCondEnt}) we get that
\begin{align*}
I(X,Y;Z) &= H(X,Y) - H(X,Y|Z)\\
&=H(X)+H(Y|X)-H(X|Z)-H(Y|X,Z)\\
&=I(X;Z)+I(Y;Z|X).
\end{align*}
\end{proof}

Consider three random variables $Y_1,Y_2,Z$, where $Y_1$ and $Y_2$ are conditionally independent given $Z$. That is,  for every $y_1,y_2,z$ such that $\Pr[Z=z]>0$ we have
$$\Pr[Y_1=y_1 \wedge Y_2 =y_2|Z=z]=\Pr[Y_1=y_1 |Z=z] \cdot \Pr[ Y_2 =y_2|Z=z].$$
The following lemma shows that the amount of information that $(Y_1,Y_2)$ give about $Z$, is at most the amount that $Y_1$ gives on $Z$ plus the amount that $Y_2$ gives on $Z$. (This is not necessarily true without the conditionally independent assumption.)

\begin{lemma}
Let $Y_1,Y_2,Z$ be random variables, where $Y_1$ and $Y_2$ are conditionally independent given $Z$. Then,
$$
I(Z;Y_1) + I(Z;Y_2) \geq I(Z; Y_1,Y_2).
$$
\end{lemma}

\begin{proof}
By repeated application of the chain rule of mutual information (\cref{claim:chainMutualInfo}), it holds
\begin{align*}
I(Z; Y_1,Y_2) &= I(Z;Y_1) + I(Z;Y_2|Y_1)\\
&= I(Z;Y_1) + I(Z,Y_1;Y_2) - I(Y_1;Y_2)\\
&=I(Z;Y_1)+I(Z;Y_2)+I(Y_1;Y_2|Z)-I(Y_1;Y_2)\\
&=I(Z;Y_1)+I(Z;Y_2)-I(Y_1;Y_2)\\
&\leq I(Z;Y_1)+I(Z;Y_2),
\end{align*}
where the last equality follows since $I(Y_1;Y_2|Z)=0$ as $Y_1$ and $Y_2$ are conditionally independent.
\end{proof}

The following lemma shows that if $I(X;Y|Z)$ is high and if $H(Z)$ is low, then $I(X;Y)$ must also be high. That is, if $X$ gives a lot of information on $Y$ when conditioning on a random variable $Z$ with low entropy, then $X$ gives a lot of information on $Y$ even without conditioning on $Z$.

\begin{lemma}
Let $X,Y,Z$ be three random variables. Then,
$$I(X;Y)\geq I(X;Y|Z)-H(Z).$$
\end{lemma}

\begin{proof}
As in the previous proof, using the chain rule of mutual information (\cref{claim:chainMutualInfo}) we have that
\begin{align*}
I(Z; X,Y) &= I(Z;X) + I(Z;Y|X)\\
&= I(Z;X) + I(Z,X;Y) - I(X;Y)\\
&=I(Z;X)+I(Z;Y)+I(X;Y|Z)-I(X;Y).
\end{align*}
Therefore,
\begin{align*}
I(X;Y) &= I(Z;X)+I(Z;Y)+I(X;Y|Z)-I(Z; X,Y)\\
&\geq I(X;Y|Z)-I(Z; X,Y)\\
&\geq I(X;Y|Z)-I(Z; Z)\\
&= I(X;Y|Z)-H(Z).
\end{align*}
\end{proof}

\section{Additional Preliminaries from Differential Privacy}
\label{app:additional}

The following theorem bounds the mutual information between the input and the output of a differentially private algorithm (that operates on a database of size 1).

\begin{theorem}[\cite{BassilyS15}]\label{thm:DPinformation}
Let $X$ be uniformly distributed over $\calX$. Let $\calA$ be an $(\epsilon,\delta)$-differentially private algorithm that operates on a single input (i.e., a database of size 1) from $\calX$. Let $Z$ denote $\calA(X)$. Then,
$$
I(X;Z)=O\left( \epsilon^2 + \frac{\delta}{\epsilon}\log|\calX| + \frac{\delta}{\epsilon}\log(\epsilon/\delta)  \right).
$$
\end{theorem}

In our protocols we will use the following protocol in the local model for computing histograms.

\begin{theorem}[Histogram protocol~\cite{BassilyS15,BNST17,BunNS19}]
\label{thm:LDPhistograms}
Let $\beta,\epsilon \leq 1$ and $\calX$ be some finite domain. There exists a 1-round $(\epsilon,0)$-differentially private protocol in the local model for $n$ parties with message complexity 1, in which the input of each agent is a single element from $\calX$ and the outcome is a data structure $D:\calX\rightarrow[n]$ 
such that for every input to the protocol $\vec{x}\in \calX^n$, with probability at least 
$1 - \beta$, for every input vector $x=(x_1,\dots,x_n)\in\calX$ we have
$$
\Big|\;D(x)-\left|\{i:x_i=x\}\right|\;\Big|\leq  O\left(\frac{1}{\eps}\cdot\sqrt{n\cdot\log\left(\frac{|\calX|}{\beta}\right)}\right).
$$
\end{theorem}

We next recall the sub-sampling technique from~\cite{KLNRS11,BBKN12}.

\begin{theorem}[Sub-sampling~\cite{KLNRS11,BBKN12}]\label{clm:boostPrivacy}
Let $\calA_1$ be an $(\epsilon^*,\delta)$-differentially private algorithm operating on databases of size $n$.
Fix $\epsilon\leq 1$, and denote $t=\frac{n}{\epsilon}(3+\exp(\epsilon^*))$.
Construct an algorithm $\calA_2$ that on input a database $D=(z_i)_{i=1}^t$ 
uniformly at random selects a subset $T\subseteq\{1,2,...,t\}$ of size $n$, and runs $\calA_1$ on the multiset 
$D_T=(z_i)_{i\in T}$.
Then, $\calA_2$ is $\left(\epsilon,\frac{4\epsilon}{3+\exp(\epsilon^*)}\delta\right)$-differentially private.
\end{theorem}

\paragraph{Secure addition protocols in the shuffle model.}

Ishai et al.~\cite{IKOS06} gave a protocol where $n\geq 2$ parties communicate with an analyzer (as in \cref{rem:model}) to compute the sum of their inputs in a finite group $G$, in the semi-honest setting and in the presence of a coalition including the analyzer and up to $n-1$ parties. In their protocol, each participating party splits their input into $\ell=O(\log|G| + \log n + \sigma)$ shares and sends each share in a separate message through the shuffle. Upon receiving the $n\ell$ shuffled messages, the analyzer adds them up (in $G$) to compute the sum. Recent work by Ghazi et al.~\cite{GhaziMPV20} and Balle et al.~\cite{BalleBGN20} improved the dependency of the number of messages on the number of participating parties to $\ell=O\left(1+(\log|G| + \sigma)/\log n\right)$.

\begin{theorem}[\cite{IKOS06, GhaziMPV20, BalleBGN20}]
\label{thm:IKOS} Let $G$ be a finite group. There exist a one-round shuffle model summation protocol with $n$ parties holding inputs $x_i\in G$ and an analyzer. The protocol is secure in the semi-honest model, and in the presence of coalitions including the analyzer and up to $n-1$ parties. 
\end{theorem}

\end{document}